\pdfoutput=1
\documentclass[acmlarge]{acmart}
\usepackage{bbm,nicefrac,bm}
\usepackage{fullpage}
\usepackage{xcolor}
\usepackage{graphicx}
\usepackage{subfig}
\usepackage{multirow}
\usepackage{abraces}
\usepackage{booktabs}
\usepackage{url}
\usepackage[abs]{overpic}
\usepackage{pdflscape}
\usepackage{afterpage}

\usepackage{tikz}
\tikzstyle myBG=[line width=3pt,opacity=1]

\newcommand{\drawPolarLinewithBG}[2]
{
  \draw[white,myBG]  (#1) -- (#2);
  \draw[black,very thick] (#1) -- (#2);
}

\newcommand{\stargraph}[2]{
    \node[circle,fill=black] at (360:0mm) (center) {};
    \foreach \n in {1,...,#1}{
        \node[circle,fill=black] at ({\n*360/#1}:#2cm) (n\n) {};
        \draw[very thick] (center)--(n\n);
        \node at (0,-#2*1.5) {}; 
    }
}

\usetikzlibrary{arrows,automata,positioning}
\tikzset{edge/.style = {->,> = latex'}}
\usepackage{tikz-qtree}
\tikzset{
      treenode/.style = {align=center, inner sep=0pt, text centered,
        font=\sffamily},
      arn_n/.style = {treenode, circle,black, draw=black, text width=1.5em}
    }
\usepackage{bbold}
\usetikzlibrary{decorations.markings}
\tikzstyle{vertex}=[circle, draw, inner sep=0pt, minimum size=6pt]

\DeclareMathOperator{\trace}{trace}
\DeclareMathOperator{\DLC}{DLC}
\DeclareMathOperator{\nDLC}{nDLC}

\newcommand{\set}[1]{\left\{ #1 \right\}}
\newcommand{\paren}[1]{\left(#1 \right)}
\newcommand{\one}{\mathbbm{1}}
\newcommand{\norm}[1]{\left\| #1 \right\|}
\newcommand{\R}{\mathbb{R}}
\newcommand{\half}{\nicefrac{1}{2}}
\newcommand{\tup}[1]{\left\langle #1 \right\rangle}
\newcommand{\tHLM}{THeLMa}

\theoremstyle{plain}
\newtheorem{theorem}{Theorem}
\newtheorem{lemma}[theorem]{Lemma}
\theoremstyle{definition}
\newtheorem*{define}{Definition}

\newcommand{\Remove}[1]{}
\colorlet{GREEN}{green} 
\colorlet{BLACK}{black}

\graphicspath{ {./} }

\usepackage{soul}

\title{Directional Laplacian Centrality for Cyber Situational Awareness}
\author{Sinan G.~Aksoy}
\email{sinan.aksoy@pnnl.gov}
\author{Emilie Purvine}
\email{emilie.purvine@pnnl.gov}
\author{Stephen J.~Young}
\email{stephen.young@pnnl.gov}
\affiliation{~Pacific Northwest National Laboratory}


\begin{document}
\begin{abstract}
Cyber operations is drowning in diverse, high-volume, multi-source data. In order to get a full picture of current operations and identify malicious events and actors analysts must see through data generated by a mix of human activity and benign automated processes. Although many monitoring and alert systems exist, they typically use signature-based detection methods. We introduce a general method rooted in {\it spectral graph theory} to discover patterns and anomalies without {\it a priori} knowledge of signatures. We derive and propose a new graph-theoretic centrality measure based on the derivative of the graph Laplacian matrix in the direction of a vertex. To build intuition about our measure we show how it identifies the most central vertices in standard network data sets and compare to other graph centrality measures. Finally, we focus our attention on studying its effectiveness in identifying important IP addresses in network flow data. Using both real and synthetic network flow data, we conduct several experiments to test our measure's sensitivity to two types of injected attack profiles, and show that vertices participating in injected attack profiles exhibit noticeable changes in our centrality measures, even when the injected anomalies are relatively small, and in the presence of simulated network dynamics.

\end{abstract}
\maketitle
\section{Introduction}
Cyber situational awareness comprises many tasks from understanding the threat landscape and network vulnerabilities to real-time network monitoring \cite{jajodia2009cyber}.
Endsley describes general situational awareness of dynamic systems in terms of three levels of awareness: perception, comprehension, and projection (or prediction) \cite{endsley1995toward}.
In order to achieve any of perception, comprehension or prediction analysts must navigate a variety of data sources and a deluge of data within each source in order to get an understanding of current operational posture.
Fusing information derived from these data sources in a way that is interpretable to a human analyst provides a holistic situational awareness.
In this paper we focus on a specific aspect of perception: network anomaly attribution.

One of the difficult problems of perception in cyber systems is the ability to see through real-time network traffic data produced by benign automated processes and normal human activity to pick out abnormal, potentially malicious, events and actors.
Network traffic data is often messy and massive, requiring assistance from algorithms or analytics in order to provide the human analyst with tips and cues for further investigation.
Many of the existing monitoring and alert systems use signature-based detection techniques, and are highly tailored to specific data types.
In order to detect zero-day attacks and other previously unseen tactics, techniques, and procedures a signature-agnostic anomaly detection strategy must be utilized.
Broadly speaking, two driving challenges are: (1) identifying the presence of anomalous or adversarial behavior within the normal background variation of network data; and (2) identifying which actors or agents within a  system are participating in the anomalous behavior. 
In this regard, a {\it mathematically grounded} approach is helpful to explore data and discover patterns and anomalies without prior knowledge of behaviors of interest.
While there are approaches that attempt to simultaneously address both challenges~\cite{nevat2017anomaly,pitropakis2018enhanced}, in this work we focus on developing a framework for addressing anomaly attribution using a method based on the {\it spectrum} of the network's graph structure. 
In addition to these two driving challenges there is also the question of determining maliciousness of an anomaly. 
The work we propose here is not meant to answer this question but rather to identify the actors or agents participating in anomalous behavior. 
Subsequent packet inspection and other investigation activities would be necessary to determine the maliciousness of the anomaly.

As a proxy for network structure, and to provide grounding for this work in commonly available network data, we will focus on network flow, a summary of computer to computer communications across the network. 
The method we propose is applicable to any data with records of the form ({\it source, destination, time, metadata}), not just within the cyber domain, although interpretation of the method's output will vary by data type.
Our approach uses a graph to model a set of network flow records, typically from a small time window.
In a spectral approach the structural properties of network communication, as reflected in a graph derived from flow, are analyzed using eigenvalues of associated matrices; these eigenvalues are known as the {\it graph spectrum}. 
The temporal sequence of flow data is then analyzed using a sequence of small time windows, each on the order of a few minutes. 
The spectrum of each small time window is computed separately providing a numerical graph spectrum vector for each window.
When stitched together over time this creates a timeline or heartbeat of the network flow data that can be inspected by an analyst to look for patterns and deviations from those patterns. 
See Figure \ref{fig:ex_timeline} for an example visualization of synthetic network flow with a scan-like anomaly around minute 30.
While this tool has shown promise in helping to identify ``low and slow'' scans (that is, scans that are not high volume all at once, but rather distributed across a period of time) as well as frequent recurrent behavior in operational data, the process requires a human in the loop to inspect the timeline and identify the actor responsible for the spectral anomaly, which can be error-prone.

\begin{figure}
\centering
\includegraphics[scale=0.5]{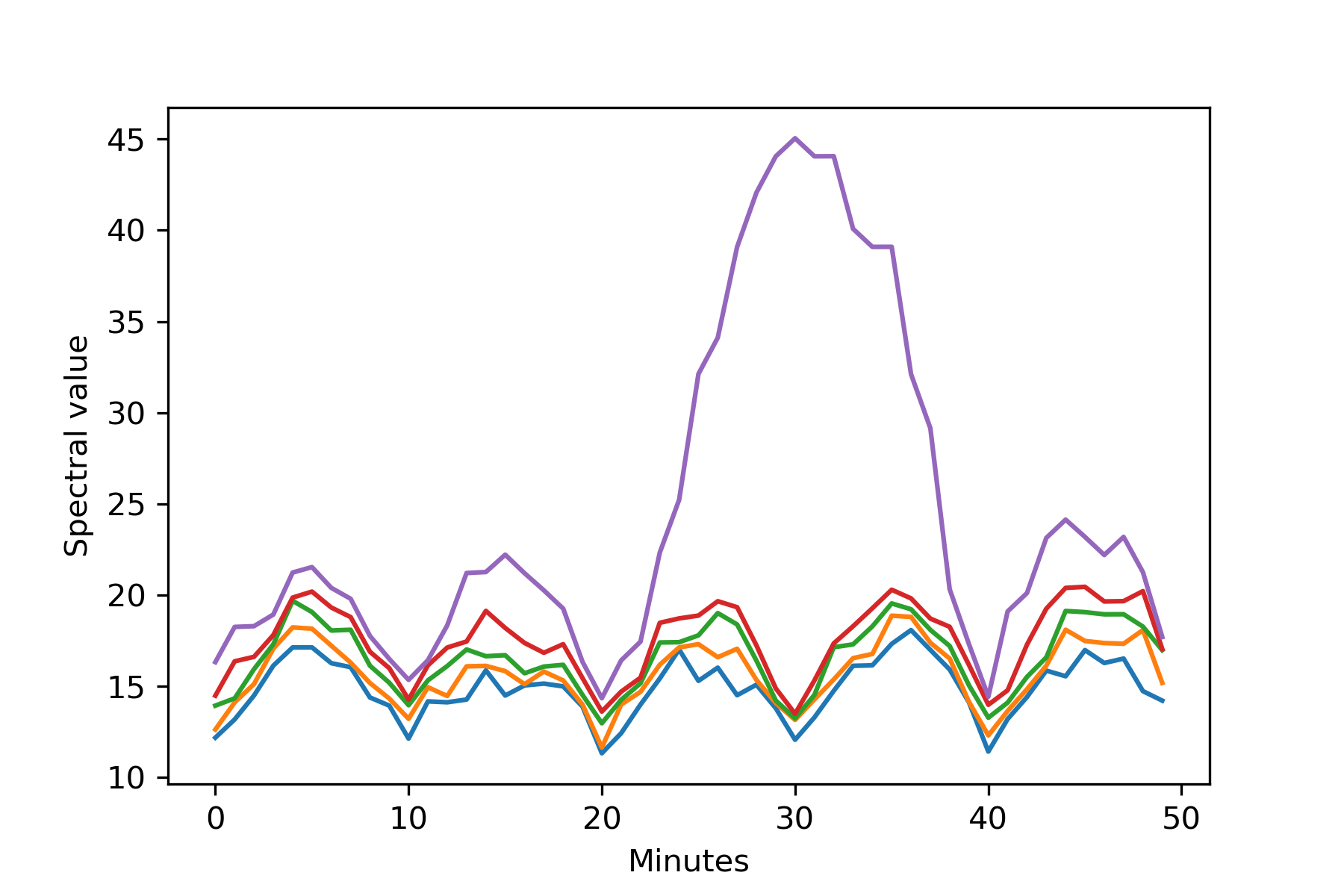}
\caption{Example spectral timeline with scan-like anomaly around minute 30.}
\label{fig:ex_timeline}
\end{figure}

This paper proposes a tool that we call the Directional Laplacian Centrality to identify ``important'' Internet Protocol (IP) addresses within each time window and focus analyst attention from the broad temporal awareness that the spectrum itself may provide to more specific attribution of anomalies.
The underlying philosophy is that an IP is important if it has a significant (measurable) contribution to the spectrum of the associated graph.
Anomalies can then be characterized as times in which a previously unimportant IP becomes significantly more important. 
After defining the Directional Laplacian Centrality we will test the method within real and synthetic data, analyzing how the importance of IP addresses change in response to planted attack profiles. 

We organize our work as follows. In Section \ref{sec:prelims}, we provide the necessary preliminaries to introduce network flow data and spectral graph theory, justify the advantages of taking a spectral approach, and provide a literature review of other relevant importance measures. 
In Section \ref{sec:specImp}, we derive and define our proposed importance measure, Directional Laplacian Centrality.
In Section \ref{sec:applications}, we provide some intuition through examples of Directional Laplacian Centrality applied to well-known graph data sets.
In Section \ref{sec:exp}, we describe our specific test data, experimental setup, and report our findings. 

\section{Preliminaries}\label{sec:prelims}
A graph $G=(V,E)$ is a set $V$ of elements, called {\it vertices}, and a set $E$ of unordered pairs of vertices, called {\it edges}.
If $\{u,v\}\in E$, we say $u$ and $v$ are {\it adjacent} and write $u \sim v$. 
We call $\{v \in V: u \sim v\}$ the {\it neighborhood of $u$}, and $d_u=|\{v \in V: u \sim v\}|$ the {\it degree} of $u$. 
A set of vertices $S \subseteq V$ is {\it connected} if for any $u,v \in S$, there exists a sequence of adjacent vertices $u,\dots,v$. 
A maximal connected subset (one which is not strictly contained within a larger subset) is called a {\it connected component} and its {\it size} is the number of vertices in the component, $|S|$.
For other basic graph theory terminology, readers are referred to \cite{bondy1976graph}. 

In this work, for ease of exposition and simplicity, we focus on unweighted, undirected graphs. 
In most applications, however, the edges of the graph may have natural weight or directionality information. 
The generalizations of the techniques developed in Section \ref{sec:specImp} to weighted and directed cases (as appropriate) is a straightforward exercise. 

\subsection{Network Flow Data}\label{sec:data}
Network flow summarizes data exchanged between pairs of IP addresses in a network.
A single network flow is an aggregation of multiple packets that occur within a small time window and have the same source IP, destination IP, source port, destination port, and protocol.
There is significant subtlety and engineering that goes into packet capture and flow aggregation (for example, even identifying which IP is source and which is destination is not trivial), as seen in one of the seminal papers \cite{claffy1995parameterizable} and the more recent \cite{hofstede2014flow}.
However, for this paper it is enough to think of a single flow as a record that contains a source and destination IP, and a time stamp. 
Other metadata are present in a record (e.g., bytes, packets, ports, and protocol) which one could use to weight or filter the data.
Initially created for accounting and usage profiling, network flow is a prevalent source of data for cyber situational awareness. Advances in flow monitoring are continuing to enable this usage \cite{jirsik_cyber_2020}. 
A recent paper of Moustafa, Hu, and Slay provides a thorough survey of network anomaly detection approaches using the paradigm of feature identification followed by classification, clustering, machine learning, rule-based, or statistical methods \cite{moustafa2019holistic}.
Our approach fits this paradigm with features derived from a graph interpretation of the network flow data followed by a statistical anomaly detection based on those derived numerical features.

A set of network flow can be modeled as a graph in which vertices represent IP addresses and an edge $\{u, v\}$ indicates that a flow between IPs $u$ and $v$ is contained in the set.
Although one could restrict to network flows with specific ports or protocols for a more targeted analysis, in our experiments we will not filter or label vertices and edges by ports or protocols but will consider all flows.
To construct each graph we restrict to a specific time interval in order to capture the structure of communications within a given time period. However,
network flow gathered from even a few hours typically creates a graph much too large to yield helpful insights, so in order to represent time and capture structural evolution we consider a dynamic graph model. 
A {\it dynamic graph} is a sequence of graphs $\{G_t\}_{t=1}^N$ where $t$ denotes a specific time, either instantaneous or a short time window, and $G_t = (V_t, E_t)$.
Dynamic graphs in which $t$ represents an instantaneous time can be difficult to analyze since at any given instance the graph may be quite sparse and fragmented; and in the case of network flow, they can evolve very rapidly as some flows have duration on the order of milliseconds.
Instead we consider small time windows, on the order of 30 seconds to 10 minutes, and gather all flow records that overlap the time window into a single graph, $G_t$.
For a dynamic graph created from network flow an edge exists in $G_t$ between two IPs when a flow occurs with one as source and the other as destination within the time window indicated by $t$. 
The vertices of $G_t$ are IP addresses present in at least one flow within the time window.
Many of the examples in this paper will be on a fixed snapshot of network flow data, a single $G_t$, see Section \ref{sec:exp} for details of that graph.

\subsection{Spectral Graph Theory}\label{SS:spectral}

Our methods for identifying potential anomalous actors within networks are rooted in {\it spectral graph theory}. 
This area of mathematics studies properties of a graph through the eigenvalues and eigenvectors of matrices associated with the graph, as illustrated in Figure \ref{fig:specTheory}. 
\begin{figure}[b]
    \centering
    \includegraphics[scale=0.3]{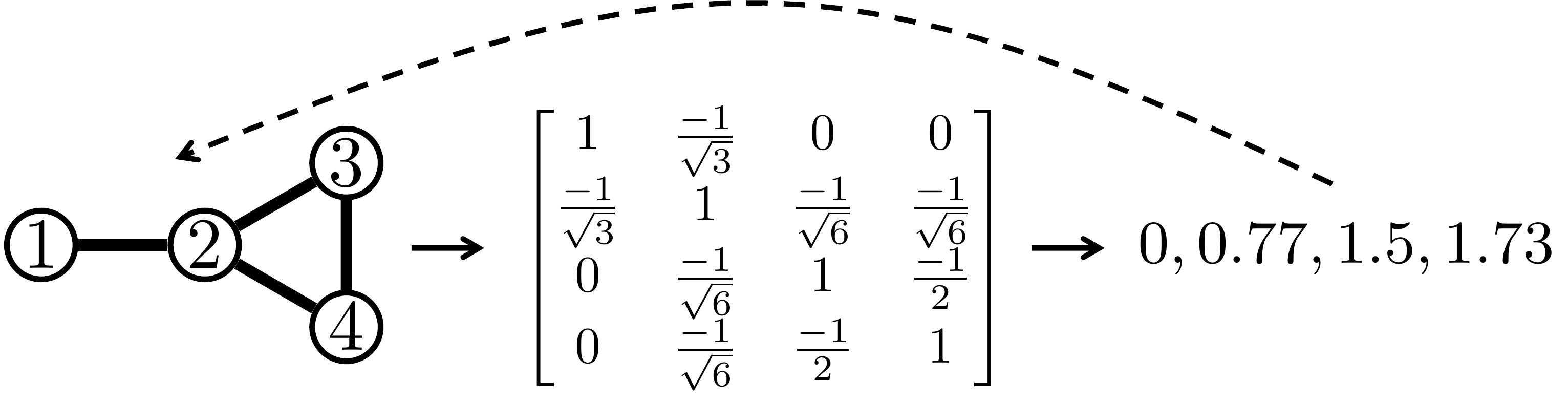}
    \caption{From left to right: a graph, its normalized Laplacian matrix, and normalized Laplacian eigenvalues. Spectral graph theory establishes relationships between eigenvalues and the graph they are derived from, represented in the diagram by the dashed line.}
    \label{fig:specTheory}
\end{figure}
There are a variety of different matrices one could associate with a graph. 
Perhaps the most basic and well-known, the adjacency matrix $A$ of an $n$-vertex graph $G=(V,E)$ is an $n \times n$ matrix where
\[
A_{ij}=\begin{cases}
1 & \mbox{ if } \{i,j\} \in E \\
0 & \mbox{ otherwise }
 \end{cases}.
\]
Two other important matrices, which will be the focus of this work, are the combinatorial Laplacian matrix, $L$ and normalized Laplacian matrix $\mathcal{L}$,
\begin{align*}
L &= D-A, \\
\mathcal{L}&=D^{-1/2}LD^{-1/2},
\end{align*}
where $A$ is the adjacency matrix defined above and $D$ is the diagonal matrix with vertex degrees on the diagonal, i.e., $D=\mbox{diag}(A{\bf 1})$, where ${\bf 1}$ is the appropriately sized all-ones vector. The matrices $A$, $L$, and $\mathcal{L}$ are all symmetric, and hence have real eigenvalues which we label in increasing order:
\[
\lambda_1 \leq \lambda_2 \leq \dots \leq \lambda_n.
\]
When unclear from context, we specify the matrix underlying a particular eigenvalue by writing, for example, $\lambda_i(L)$ denotes that $\lambda_i$ is the $i$th eigenvalue of the combinatorial Laplacian.

Unlike $A$, both $L$ and $\mathcal{L}$ are positive semi-definite, and therefore have non-negative eigenvalues. 
The eigenvalues, or {\it spectrum}, of $L$ and $\mathcal{L}$ characterize a number of important properties not captured by the adjacency eigenvalues. Particularly pertinent to our analyses, Laplacian spectra capture various graph connectivity and neighborhood expansion properties. Loosely speaking, such properties quantify the extent to which a graph (or subsets of vertices within a graph) are well-connected.  
For instance, an elementary fact is that the multiplicity of the eigenvalue 0 of both $L$ and $\mathcal{L}$ equals the number of connected components of the graph. 
Furthermore, the second eigenvalue $\lambda_2$ of a connected graph quantifies the extent to which that graph is well-connected, in several different regards.  

In the case of the combinatorial Laplacian $L$, this second eigenvalue is referred to as {\it algebraic connectivity}. 
This quantity is related to the graph's vertex connectivity, the minimum number of vertices that must be deleted to disconnect a graph.  
Furthermore, the corresponding eigenvector, called Fiedler's vector, is frequently used to partition a graph into well-connected groups; see \cite{fiedler1975property} for more. 
In the case of the normalized Laplacian, the second eigenvalue captures neighborhood expansion properties, which measure how many edges leave a set of vertices, relative to the ``volume" of that set. 
One way of quantifying such expansion is via Cheeger's constant. 
More precisely, for a vertex subset $X \subseteq V$, letting $e(X,\overline{X})$ denote the number of edges between $X$ and its complement $\overline{X}$ and defining the volume of a set as $\mbox{vol}(X)=\sum_{i \in X}d_i$, the {\it Cheeger ratio} of $X$ is 
\[
\Phi(X)=\frac{e(X,\overline{X})}{\min\{\mbox{vol}(X),\mbox{vol}(\overline{X})\}}.
\]
The {\it Cheeger constant} of the graph $\Phi(G)$ is the minimum Cheeger ratio over all vertex subsets. 
Via Cheeger's inequality \cite{chung1997spectral}, $\lambda_2(\mathcal{L})$ serves as an approximation of the Cheeger's constant,
\[
2\Phi(G) \geq \lambda_2(\mathcal{L}) \geq \frac{\Phi(G)^2}{2}.
\]
In this sense, the second eigenvalue controls graph neighborhood expansion properties. Further, just as the aforementioned Fiedler vector may be used to partition a graph, the normalized Laplacian eigenvector corresponding to $\lambda_2(\mathcal{L})$ can be used to partition the graph as well. 

While the aforementioned connectivity properties were key to the early development of spectral graph theory, researchers have shown a wide variety of graph structural properties are encoded in eigenvalues. 
For example, vertex degrees \cite{das2005sharp}, average shortest path length \cite{mohar1991eigenvalues}, diameter \cite{chung1989diameters}, chromatic number \cite{wilf1967eigenvalues}, independence number \cite{godsil2008eigenvalue}, number of spanning trees \cite{boesch1986spanning}, network flows\footnote{The term ``network flow'' here is different than network flow data used in this paper. Here it refers to quantities flowing through a weighted and directed graph.} and routing \cite{alon1994routing}, and random walk mixing time and associated parameters \cite{aldous1995reversible}, can all be bounded, controlled, or characterized in terms of eigenvalues. 
This body of research attests to eigenvalues of the aforementioned matrices serving as a far-ranging tool for capturing graph properties. 

\subsection{Importance Measures}
One approach towards identifying agents participating in anomalous behavior within a graph is to understand their role within the graph structure relative to other agents and how that role changes over time. 
In this context, graph-theoretic importance measures provide different ways of quantifying an element's role in the system. Tracking the importance over time, through a dynamic graph sequence, can identify anomalous changes in behavior or role. 
Here, we briefly review examples of importance measures to elucidate the breadth of properties they can capture, and to place our proposed importance measure within the literature. 

Broadly speaking, one class of graph importance measures are based on the shortest-path structure in the graph. 
For example {\it closeness centrality} ranks vertices based on their average shortest path length to other vertices in the graph: vertices in close proximity to the rest of the graph via shorter paths are ranked more highly than vertices on the periphery. 
{\it Betweeness centrality} measures the importance of a vertex based on the frequency of its occurrence {\it en route} between other pairs of vertices. 
Accordingly, vertices that belong to many short paths linking other pairs of vertices are ranked highly. 
Rather than restrict attention to {\it shortest} paths, another related class of importance measures are based more generally on {\it walks}, which may be non-minimal in length. 
Since a pair of vertices may be linked by infinitely many walks of arbitrary length, such measures tend to rely on asymptotic expressions, such as the limiting distribution of a random walk. 
For instance, the well-known PageRank algorithm assigns scores to vertices based on the stationary distribution of a modified random walk. 
Katz centrality considers the total number of walks from a given vertex to the rest of the graph, penalizing longer walks with an attenuation factor. 

\begin{figure}[h]
    \centering
    \includegraphics[scale=0.22]{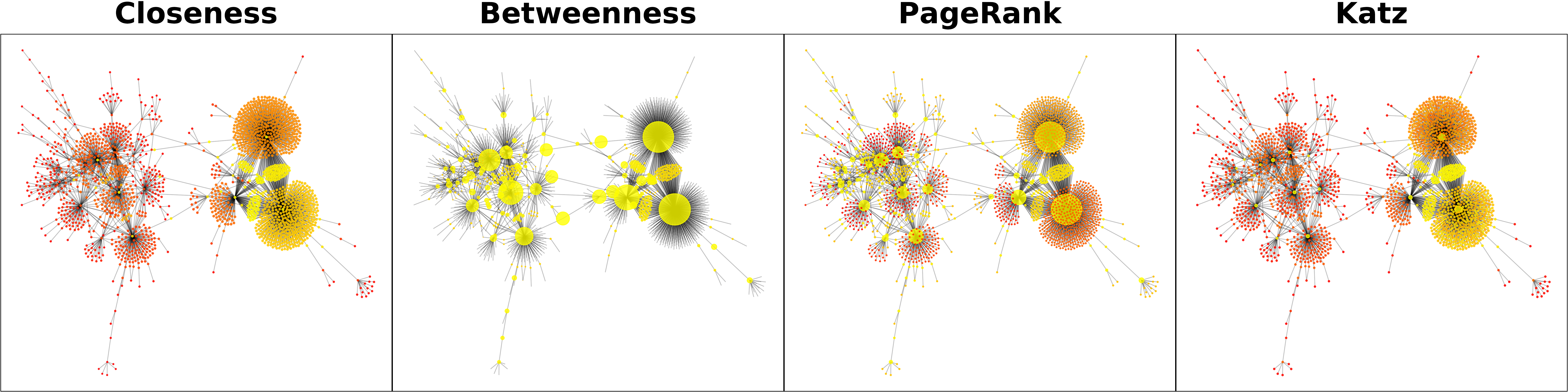}
    \caption{Different centrality scores on the same graph. Vertex size is proportional to a normalized centrality score, and color is proportional to score percentile, with yellow for higher percentiles and red for lower percentiles.}
    \label{fig:centEx}
\end{figure}

Figure \ref{fig:centEx} illustrates these different centrality scores on a single $G_t$ from a dynamic graph of network flow. 
The specific source and snapshot of the data is given in Section \ref{sec:exp} where we describe our experiments in detail but the structure is as outlined in Section \ref{sec:data}.
In Figure \ref{fig:centEx} vertices are sized proportionally\footnote{To make the vertex sizes comparable across centrality measures, the scores within each centrality measure are divided by the sum of all the scores for that centrality measure.} to their centrality score, and colored according to that score's percentile in the graph.
Comparing them, we see how seemingly subtle differences in centrality definitions can yield drastically different scores. 
For betweeness centrality, we observe a heavily skewed score distribution. 
This is partly due to the prevalence of pendant (i.e. degree 1 vertices), which cannot be on a shortest path between two other vertices, and hence receive a betweeness score of 0. 
In contrast, the closeness centrality scores are much more uniform across the graph. 
The high-degree or ``hub'' vertices are not well-distinguished from the pendant vertices because their distances, on average, to the rest of the graph are comparable. 
Turning our attention to the two walk-based measures, PageRank and Katz centrality\footnote{Here, PageRank is computed with damping parameter $\alpha=0.85$, and Katz centrality is computed with attenuation parameter $\alpha=\frac{1}{2\lambda_{\max}}$ where $\lambda_{\max}$ is the adjacency matrix spectral radius of $G$.}, we see that both better distinguish hub from pendant vertices. 
PageRank assigns relatively larger scores to hubs than Katz centrality. 
Taking a broader view, this example illustrates the importance of understanding the foundations of any proposed centrality metric in order to interpret its behavior. 

In Section \ref{sec:specImp}, we will propose and justify our own centrality score based in spectral graph theory. 
While the aforementioned centrality measures are implicitly related to eigenvalues, via the relationship of eigenvalues to random walks and walk-counting noted above, the measures we will consider are {\it explicitly} derived from eigenvalues and eigenvectors. 
In this vein, our work will build off prior research by Qi \cite{qi2012laplacian}, who defined a spectral importance measure they call {\it Laplacian centrality}. 
There, Qi defines vertex importance based on its contribution to the entire spectrum, as measured by the {\it Laplacian energy}:
\[
\mathcal{E}(G)=\sum_{i=1}^n \lambda_i^2.
\]
More precisely, the Laplacian centrality of a vertex $v$ is the percent change in Laplacian energy from deleting $v$, i.e.
\[
\frac{\mathcal{E}(G)-\mathcal{E}(G\setminus v)}{\mathcal{E}(G)},
\]
where $G\setminus v$ denotes the graph formed by deleting $v$ and all of its edges from $G$. 
We note Laplacian centrality is always non-negative \cite{qi2012laplacian}. 
We will return to and visualize Laplacian centrality alongside our proposed notion of centrality at the end of Section \ref{sec:specImp}.

\section{Directional Centrality Measures} \label{sec:specImp}

As we have noted in Section \ref{SS:spectral} functions of the various spectra of the graph can encode a variety combinatorial properties of the underlying graph. 
This motivates using spectral information to measure the importance of particular vertices within that structure. 
Perhaps the most natural way to evaluate the importance of a particular vertex to a spectral function is the aforementioned approach taken by Qi et al: remove the vertex from the graph, recalculate the spectrum and the associated function and record the change.

However, in some cases deleting a vertex from a graph may induce dramatic changes to the graph structure that disproportionately affect the spectrum. 
For example, recall from Section \ref{SS:spectral} that the Cheeger constant, a measure of neighborhood expansion, is closely tied to the second eigenvalue of the normalized Laplacian. 
Consequently, the deletion of {\it any} vertex that disconnects a graph -- including, for example, vertices whose removal only isolates a single vertex from the rest of the graph -- converts the Cheeger constant and second eigenvalue to 0.
Other graph parameters we surveyed in Section \ref{SS:spectral} may also be changed significantly by the deletion of a vertex that, in turn, induces jump-shifts in the spectrum. 
For this reason, assessing importance based on outright vertex deletion outcomes may be an insufficiently nuanced approach. 
Instead, we propose measuring importance using a function of the spectrum based on an {\it infinitesimal} change in the graph structure. More formally, we consider the derivative of an eigenvalue in the direction of a vertex. 
In the remainder of this section, we develop this notion and use it to define a more targeted notion of Laplacian centrality that we call {\it Directional Laplacian Centrality}. 
Following this technical derivation we will build overall intuition for the behavior of Directed Laplacian Centrality by applying it to a variety of well-studied networks in Section \ref{sec:applications}.  We will then perform experiments using real and synthetic network flow data in Section \ref{sec:exp} to show how the Directional Laplacian Centrality changes in the presence of two types of planted anomalies and indicate how it could be used by an analyst to facilitate situational awareness.

\subsection{Eigenvalue Directional Derivative}\label{SS:eigen}
Before we develop the notion of derivative of an eigenvalue in the direction of a vertex, we first recall some basic facts about the dependence of an eigenpair on the entries of a matrix, (see for instance \cite{Magnus:Eigendiff}). 

\begin{theorem}
Let $X_0$ be an real-symmetric matrix in $\R^{n \times n}$ and let $(\lambda_0,v_0)$ be an associated eigenpair such that $\norm{v_0} = 1$.  If $\lambda_0$ is a simple eigenvalue then there is a neighborhood  $N(X_0)$ and functions $\lambda \colon N(X_0) \rightarrow \R$ and $v \colon N(X_0) \rightarrow \R^n$ such that for all $X \in N(X_0)$
\begin{enumerate}
    \item $\lambda(X_0) = \lambda_0$,
    \item $v(X_0) = v_0$,
    \item $\norm{v(X)} = 1$, and 
    \item $Xv(X) = \lambda(X)v(X).$
\end{enumerate}
Furthermore, $\lambda$ and $v$ are infinitely differentiable on $N(X_0)$.  
\end{theorem}

In particular, this implies that if $(\lambda,v)$ is a simple eigenpair associated with a real-symmetric matrix $A \in \R^{n\times n}$, then for any other matrix $B \in \R^{n\times n}$, the matrix function $M(t) = A + tB$ defines, via the implicit function theorem, a pair of functions $(\lambda(t),v(t))$ which are  infinitely differentiable in a neighborhood, $N$, of $0$ and such that $\norm{v(t)} = 1$ and $M(t)v(t) = \lambda(t)v(t)$  for all $t \in N$.  If we think of $A$ as the adjacency matrix of a graph and $B$ as the collection of edges incident to a vertex $x$, this allows us to naturally define a notion of a directional derivative in the direction of the vertex $x$ for any simple eigenpair.  Specifically, recalling that given a parameterized matrix $M(t)$ and the associated eigenpair functions $(\lambda(t),v(t))$, we have that 
\[ \frac{d\lambda}{dt}(t_0) = v(t_0)^T \frac{dM}{dt}(t_0) v(t_0),\] 
and thus, the directional derivative in the direction of $x$ of $(\lambda,v)$ for the adjacency matrix is given by $v_x \sum_{y \sim x} v_y.$    Using this framework, we will define the directional derivative for eigenpairs of the both the combinatorial and normalized Laplacian.

\begin{lemma}\label{L:derivative}
Let $G = (V,E)$ be a simple graph and let $x$ be an arbitrary vertex.  If $(\lambda,v)$ is an eigenpair of the combinatorial Laplacian $L$ of $G$, then the derivative of $\lambda$ in the direction $x$ is given by 
\[ \sum_{y \sim x} \paren{v_x - v_y}^2.\]  If instead, $(\lambda,v)$ is an eigenpair of the normalized Laplacian of $G$, then the derivative of $\lambda$ in the direction $x$ is given by 
\[(1-\lambda) \sum_{y \sim x} \paren{ \frac{v_x}{\sqrt{d_x}} - \frac{v_y}{\sqrt{d_y}}}^2 - \lambda \sum_{y \sim x} \frac{2v_xv_y}{\sqrt{d_xd_y}}.\]
\end{lemma}

\begin{proof}
Similarly as above, we define the adjacency matrix in the direction of $x$ to be
\[ A_{ij}(t) = \begin{cases} 1 + t & \set{i,j} \in E, x \in \set{i,j} \\
1 &\set{i,j} \in E, x \not\in \set{i,j} \\
0& \textrm{otherwise} \end{cases}.\]  The matrix $D(t)$ is then the diagonal matrix of $t$-dependent degrees, i.e. the matrix which has $A(t) \one$ on the diagonal.  In particular, $D_{xx}(t) = (1+t)d_x$, $D_{yy}(t) = d_y + t$ for $y \sim x$, and $D_{zz}(t) = d_z$ for $z \not\sim x$.  Thus we have that 
\[ L_{ij}'(t) = \paren{D'(t) - A'(t)}_{ij} = \begin{cases} d_x &  i = j = x \\ 1 & i = j, \set{i,x} \in E \\ -1 & \set{i,j} \in E, x \in \set{i,j} \\ 0 & \textrm{otherwise} \end{cases},\]
and we see that the derivative of $\lambda$ in the direction $x$ is given by
\[ v(0)^T L'(0) v(0) = d_x v_x^2  + \sum_{y \sim x} v_y^2 - 2v_xv_y = \sum_{y \sim x} \paren{v_x - v_y}^2,\] as desired.

For the normalized Laplacian the situation is slightly more complicated by the need to use the non-commutative product rule, however recalling that $\mathcal{L}(t) = I - D(t)^{-\half}A(t)D(t)^{-\half}$, we have that 
\[\frac{d}{dt}\mathcal{L}(t) = \half D(t)^{-\nicefrac{3}{2}}D'(t)A(t)D(t)^{-\half}-D(t)^{-\half}A'(t)D(t)^{-\half}+\half D(t)^{-\half}A(t)D(t)^{-\nicefrac{3}{2}}D'(t).\] 
Hence, using that $D(t)^{-\nicefrac{3}{2}}D'(t) = D'(t)D(t)^{-1}D(t)^{-\half}$ and that $D(0)^{-\half}A(0)D(0)^{-\half}v = (1-\lambda)v$,
\begin{align*}
    v^T\mathcal{L}'(0)v &= \half v^T D'(0)D^{-1}(0) (1-\lambda)v - v^TD(0)^{-\half}A'(0)D(0)^{-\half} + (1-\lambda)v^T D(0)^{-1}D'(0)v \\
    &= (1-\lambda) v^T D'(0)D^{-1}(0) v - v^TD(0)^{-\half}A'(0)D(0)^{-\half} \\
    &= (1-\lambda)\paren{v_x^2 + \sum_{y \sim x} \frac{v_y^2}{d_y}} - \sum_{y \sim x} 2\frac{v_xv_y}{\sqrt{d_xd_y}} \\
    &= (1-\lambda)\paren{v_x^2  - 2\sum_{y \sim x} \frac{v_xv_y}{\sqrt{d_xd_y}} + \sum_{y \sim x} \frac{v_y^2}{d_y}} - \lambda \sum_{y \sim x} 2\frac{v_xv_y}{\sqrt{d_xd_y}} \\
    &= (1-\lambda)\sum_{y \sim x} \paren{\frac{v_x}{\sqrt{d_x}} - \frac{v_y}{\sqrt{d_y}}}^2 - \lambda \sum_{y\sim x}\frac{2v_xv_y}{\sqrt{d_xd_y}},
\end{align*}
as desired.
\end{proof}

It is easy to see that for both the combinatorial and normalized Laplacian the derivative with respect to the eigenspace corresponding to eigenvalue 0 ($\one$ for the combinatorial Laplacian and $\left\langle \sqrt{d_i}\right\rangle_i$ for the normalized Laplacian) in the direction of any vertex is zero, as expected.   Unfortunately, it is clear that if $\lambda$ is a non-simple eigenvalue (and hence corresponds to a eigenspace of dimension at least 2), then the definition of the derivative depends on the choice of eigenvector associated with $\lambda.$  However, the following result shows that if we instead define the derivative in terms of the entire eigenspace associated with $\lambda$, then the derivative is indpendent of the particular decomposition of the eigenspace.

\begin{lemma}\label{L:subspace}
Let $S$ be a $k$-dimensional subspace of $\R^n$ and let $x, y \in \R^n$.  There is some constant $C$ such that $\sum_{i} x^Tv_i y^Tv_i = C$ for any choice of orthonormal basis of $S$.
\end{lemma}
\begin{proof}
Fix an arbitrary orthonormal basis of $S$, $\set{v_1,\ldots, v_k}$ and let $P$ be an orthonormal matrix in $\R^{k\times k}$.  Define $w_i = VPe_i$ $V = [v_1 \cdots v_k]$.  Now consider
\begin{align*}
    \sum_i x^Tw_i y^Tw_i &= \sum_i x_T (VPe_i) y^T (VPe_i) \\
    &= \sum_i x^T V P e_i e_i^T P^T V^T y \\
    &= x^T V P \paren{\sum_i e_i e_i^T} P^T V^T y \\
    &= x^T V P P^T V^T y \\
    &= x^T V V^T y.
\end{align*}
As this value is independent of the choice of $P$, and in particular $P$ can be chosen to be the identity, we have the existence of the desired constant.
\end{proof}

Thus, combining the observations of Lemma \ref{L:derivative} and \ref{L:subspace}, we have the following well-founded definition.
\begin{define}[{\it Eigenspace Directional Derivative}]
Let $G = (V,E)$ be a graph and let $\lambda$ be an eigenvalue of the combinatorial Laplacian with $\set{v^{(1)},\ldots, v^{(k)}}$ an orthonormal basis of the associated eigenspace.  The derivative of $\lambda$ in the direction $x \in V$ is given by 
\[ \frac{1}{k}\sum_{i=1}^k \sum_{y \sim x} \paren{v_x^{(i)} - v_y^{(i)}}^2.\]  
If instead $\lambda$ and $\set{v^{(1)}, \ldots, v^{(k)}}$ parameterize an eigenspace of the normalized Laplacian, the derivative of $\lambda$ in the direction of $x \in V$ is
\[\frac{1}{k}\sum_{i=1}^k \paren{ (1-\lambda)\sum_{y \sim x} \paren{\frac{v^{(i)}_x}{\sqrt{d_x}} - \frac{v^{(i)}_y}{\sqrt{d_y}}}^2 - \lambda \sum_{y \sim x} \frac{2v^{(i)}_xv^{(i)}_y}{\sqrt{d_xd_y}}}.\]
\end{define}

We note that it is a straightforward exercise to generalize this definition to the case of weighted graphs or to the derivative in the direction of any subset of edges of the graph.  

\subsection{Directional Laplacian Centrality}

We also note that, while the formalism and precise definition may be new, the essential idea of the eigenspace directional derivative is solidly  rooted in prior work.  For example, the aforementioned work of Qi, et al.~\cite{qi2012laplacian},
where they define Laplacian Centrality as $\frac{\mathcal{E}(G)-\mathcal{E}(G\setminus v)}{\mathcal{E}(G)}$, where $\mathcal{E}(G)$ is the sum of squares of eigenvalues of the combinatorial Laplacian, can be understood as an attempt to capture the idea of a directional derivative.  
Specifically, letting $A$ be the adjacency matrix of $G$ and letting $B_x$ be the adjacency matrix of the edges incident to $x$, define $A(t) = A + tB_x$.  
Similarly as above, $D(t) = A(t) \one$ and $L(t) = D(t) - A(t)$.  
With this notation it is easy to see that the Laplacian Centrality of $x$ can be viewed as a re-scaling of the approximation
\[ \left.\frac{d}{dt}\trace(L^2(t))\right|_{t=0} \approx \frac{\trace(L^2(0)) - \trace(L^2(-1))}{0- (-1)}. \]
Now, given a complete orthonormal decomposition  $\set{(\lambda_i,v^{(i)})}$ for the combinatorial Laplacian, we can evaluate this derivative exactly as
\begin{align*}
\left.\frac{d}{dt}\trace(L^2(t))\right|_{t=0} &= \left.\frac{d}{dt}\paren{ \sum_i \lambda_i(t)^2}\right|_{t=0} \\
&= \sum_i 2 \lambda_i(0) \lambda'_i(0) \\
&= \sum_i 2 \lambda_i \sum_{y \sim x} \paren{v^{(i)}_x - v^{(i)}_y}^2.
\end{align*}
It is interesting to note that, by standard rearrangements of the quadratic form, $\lambda_i = \sum_{\set{x,y} \in E} \paren{v_x^{(i)} - v_y^{(i)}}^2$ and thus $\sum_{y \sim x} \paren{v_x^{(i)} - v_y^{(i)}}^2$ can be thought of as the portion of $\lambda_i$ "incident" with $x$.  

These ideas can be further expanded to understand the relative importance of any combinatorial substructure to functions of the spectrum.  
  Namely, if $f$ is a function of $\bm{\lambda} = (\lambda_1,\ldots, \lambda_k)$, the eigenspaces of a matrix associated with a graph $G = (V,E)$, and $E'$ is the set of edges associated with a combinatorial substructure of $G$, then relative importance of $E'$ to $f$ can be described by 
\[ \frac{d}{dE'}f(\bm{\lambda}) = \nabla f(\bm{\lambda})^T \frac{d\bm{\lambda}}{dE'}.\]
Returning to the Laplacian Centrality definition of Qi, et al.\ ~\cite{qi2012laplacian}, it easy to see that it can be phrased within this framework by defining $f$ as the sum over all eigenspaces $\set{(\lambda_i,V_i)}$ of $\dim(V_i)\lambda_i^2$. 

Having established this context, we introduce our proposed importance measure Directional Laplacian Centrality.

\begin{define}[{\it $S$-Directional (Normalized) Laplacian Centrality}] 
Let $G = (V,E)$ be a graph on $n$ vertices with combinatorial and normalized Laplacians $L$ and $\mathcal{L}$, respectively. Let $\set{(\lambda_i,v^{(i)})}$ be an orthonormal decomposition of the eigenstructure of $L$ (resp. $\mathcal{L}$) such that $\lambda_i \leq \lambda_{i+1}$.  For any set $S \subseteq \set{1,\ldots,n}$, the $S$-Directional Laplacian Centrality (resp. $S$-Directional Normalized Laplacian Centrality) of a vertex $x \in V$ is
\begin{align*}
    S\mbox{-}\DLC(x)&=\sum_{s \in S} \sum_{y \sim x} \paren{v_x^{(s)} - v_y^{(s)}}^2, \\
    S\mbox{-}\nDLC(x) &=
\sum_{s \in S} \paren{\paren{1-\lambda_s} \sum_{y \sim x} \paren{\frac{v_x^{(s)}}{\sqrt{d_x}} - \frac{v_y^{(s)}}{\sqrt{d_y}}}^2 - \lambda_s \sum_{y \sim x} \frac{2v_x^{(s)}v_y^{(s)}}{\sqrt{d_xd_y}}},
\end{align*}
respectively. 
\end{define}

As oftentimes we are interested in the extremal eigenvalues\footnote{This preference is motivated by theoretical observations such as the Wigner Semi-circle Law~\cite{Wigner:CollectionsI,Wigner:CollectionsII,Wigner:SemiCircle55,Wigner:SemiCircle58} and the related extensions~\cite{Tao:CircularLaw}, as well as the  spectral analysis of the random $d$-regular graph~\cite{McKay:EigenvalueDregular,Vu:EignavluesDRegular} and the Chung-Lu random graph~\cite{CLV:PowerLawEigenvalues,Chung:ExpectedDegreeSpectra,CLV:spectra} which show that behavior of the "bulk" of the spectrum of many random matrices (including those associated with random graphs) has a limiting behavior.  This indicates that if a we can think of the process that generates a graph in a "low-rank" manner, then the "bulk" of the spectrum is is attributable to stochasticity rather than intrinsic properties of the generative process.}, we use the notation $k$-$\DLC$ and $\overline{k}$-$\DLC$ to denote the sets $\set{t+1,\ldots,t+k}$ and $\set{n-k+1,\ldots, n}$ where $t$ is the dimension of the nontrivial null space of $L$.
In this way, eigenvectors in the null space are excluded because the directional derivative for such eigenvectors is always trivially zero.
As an example, we compute the Directional Laplacian Centrality and Normalized Laplacian Centrality for the first 5 and last 5 nontrivial eigenvalues of the same graph considered in Figure \ref{fig:centEx}. 
We present a visualization of the centrality values in Figure \ref{fig:LC_way}, alongside Qi's Laplacian Centrality, and Normalized Laplacian Centrality.\footnote{Although only defined for the combinatorial Laplacian by Qi in \cite{qi2012laplacian}, the Normalized Laplacian Centrality can be defined in the same way, as the percent change in the normalized Laplacian energy. }
\begin{figure}[h]
    \centering
    \subfloat[Directional Laplacian Centrality\label{fig:DLC_4way}]{\includegraphics[scale=0.28]{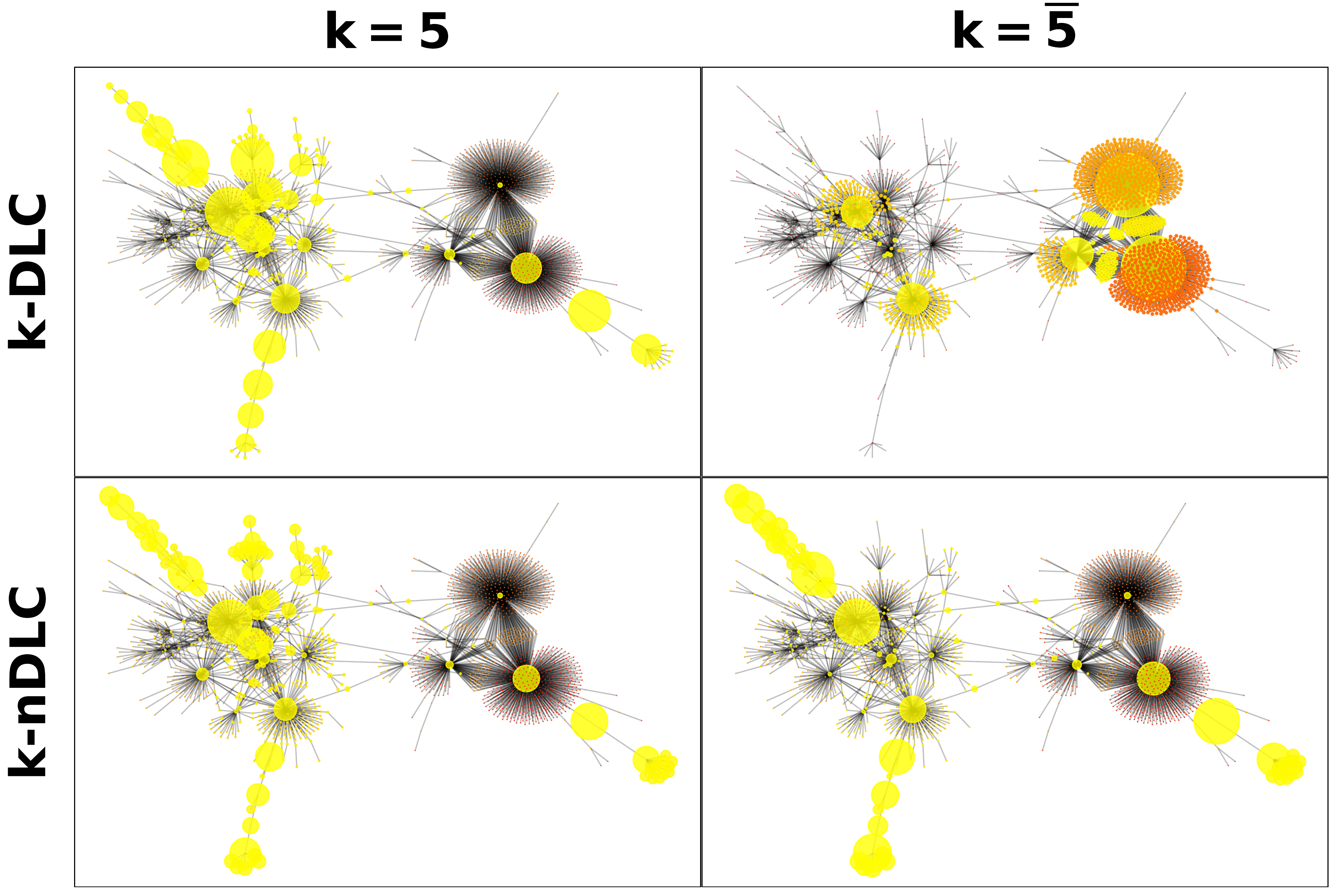}}
    \subfloat[Laplacian Centrality\label{fig:LC_2way}]{
    \includegraphics[scale=0.28]{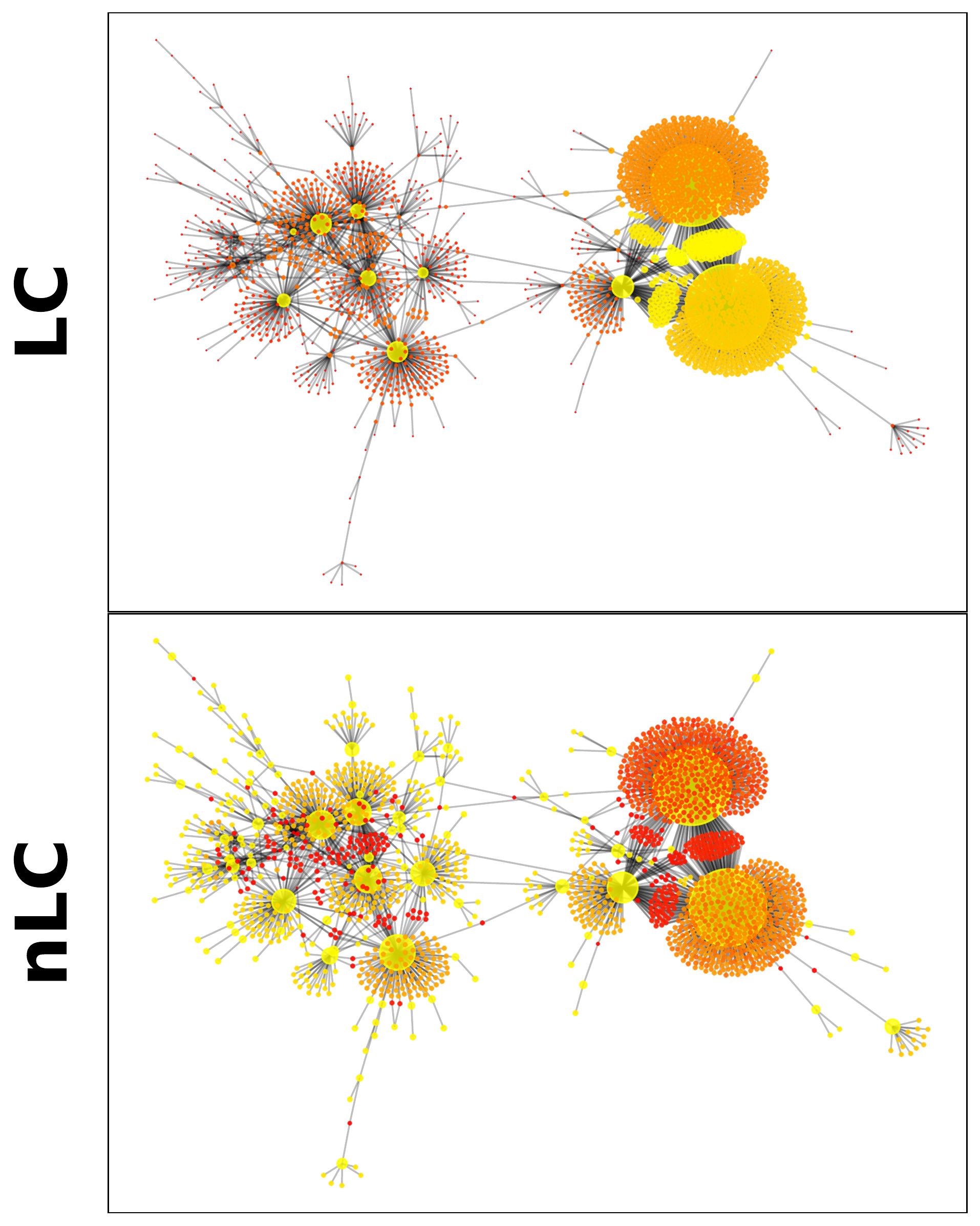}
    }
    \caption{(a): Directional Laplacian (top row) and Normalized Laplacian Centrality (bottom row) on the same graph, for the 5 largest (right side) and 5 smallest (left side) nontrivial eigenvalues. (b): Laplacian Centrality (LC) and Normalized Laplacian Centrality (nLC). Vertex size is proportional to score, and color is proportional to score percentile, with yellow for higher percentiles and red for lower percentiles.}
    \label{fig:LC_way}
\end{figure}

Comparing the visualizations within this figure, the differences underscore the impact the choice of eigenvalues and matrix may have on the scores. 
These differences in centrality score align with our understanding that the maximal and minimal eigenvalues of a graph capture very different properties. 
Comparing Figure \ref{fig:DLC_4way} to the previously depicted PageRank, Katz, Betweeness and Closeness centrality visualizations in Figure \ref{fig:centEx}, there are clear qualitative differences. 
Like those measures, DLC may ascribe importance to ``hub" vertices with high degree, particularly as seen in the $\overline{5}$-$\DLC$. 
But in certain cases, such as the $5$-$\nDLC$, low-degree vertices on long radial paths leading to hub vertices are also highlighted. 
This may be helpful when looking for lateral movement anomalies in a network as long paths can indicate such behavior.
In contrast, the distance and walk-based measures give minimal scores to these vertices, as do Qi's Laplacian Centrality and normalized Laplacian Centrality scores visualized in Figure \ref{fig:LC_2way}. 
The ability to measure structural importance beyond hubs is particularly valuable and intriguing, as this suggests DLC may reflect structural nuances otherwise lost in popular measures. 
In the next section, we conduct experiments aimed at illustrating the behavior of DLC and nDLC on a vareity of standard network science data sets.

\section{Applications of Directional Laplacian Centrality}\label{sec:applications}
Before addressing the applicability of Directional Laplacian Centrality to cyber-situational awareness in Section \ref{sec:exp}, we develop an intuition for the behavior of the $\DLC$ and $\nDLC$ by applying it to several well-understood graph data sets.  Below, we consider the Karate Club, Les Mis\'{e}rables, and Network Science collaboration data sets.  Then in Section \ref{SS:enron} we apply $\DLC$ and $\nDLC$ to the dynamic graph from the Enron email data set.

\subsection{Social Interaction Data}\label{SS:static}
\begin{figure}
    \centering
    \includegraphics[scale=0.26]{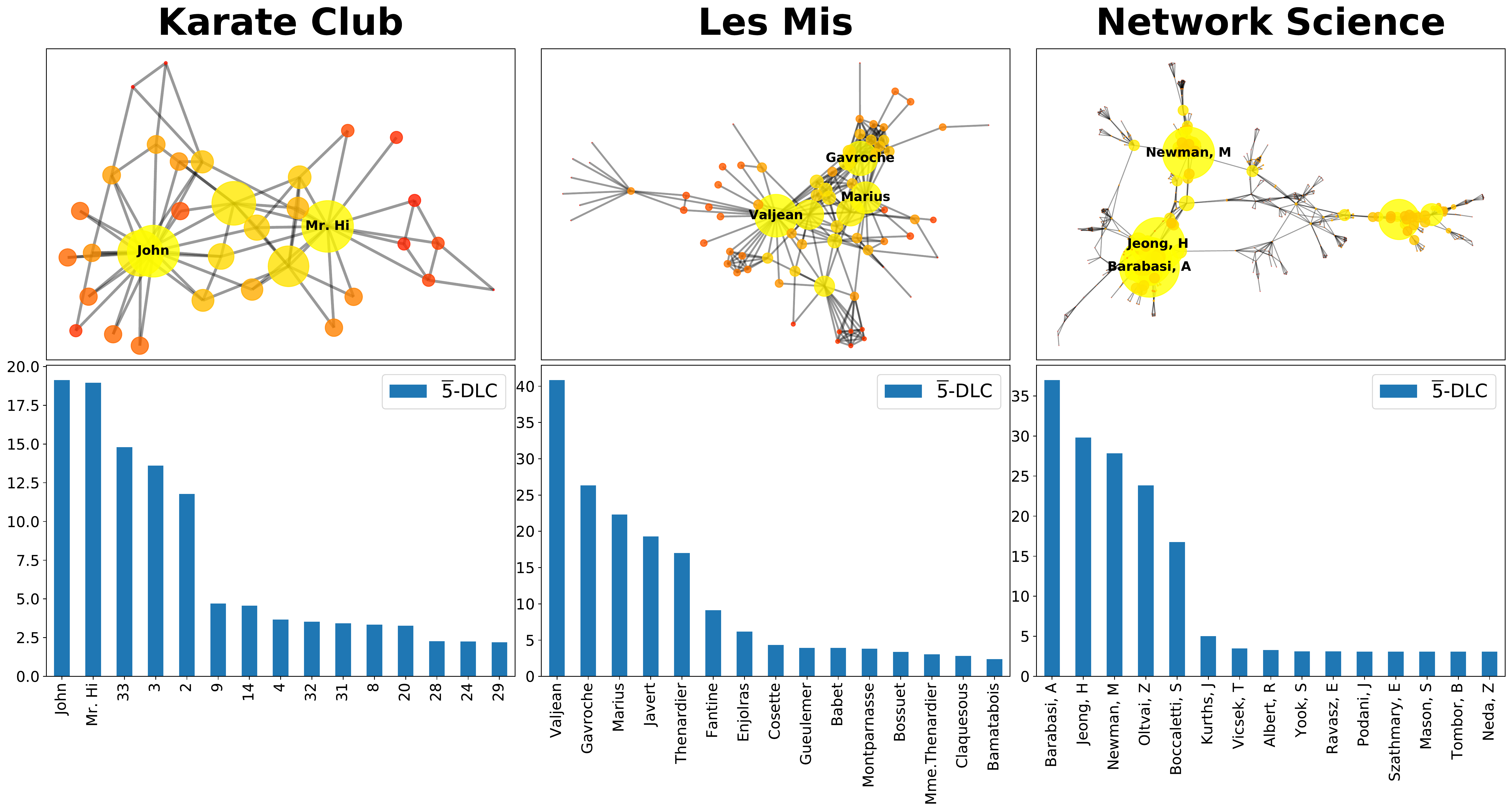}
    \caption{The Directional Laplacian Centrality scores for the 5 largest eigenvalues of the Karate Club, Les Mis and Network Science datasets. In the visualisation (top), vertex size is proportional to score, and color is proportional to score percentile, with yellow for higher percentiles and red for lower percentiles. Beneath each visualization, the plots present the top 15 scores.}
    \label{fig:socialNets}
\end{figure}
We consider several prominent data sets frequently utilized to test graph centrality methods. Each of these datasets features one or more vertices that, in their respective context, can be naturally interpreted as the most important or central to that graph. Figure \ref{fig:socialNets} presents the $\overline{5}$-DLC results for these data sets, which we describe below:
\begin{itemize}
    \item {\bf Zachary's Karate Club}. A graph of social interactions amongst members of a karate club \cite{zachary1977information}. Due to a conflict between an instructor Mr. Hi, and an administrator, John, the club fractured into two groups. The two highest scoring $\overline{5}$-DLC values are achieved by the leaders of these groups, Mr. Hi and John. 
    \item {\bf Les Mis\'{e}rables}. A co-occurrence graph from \cite{knuth1993stanford} derived from Victor Hugo's novel. Vertices represent characters which are linked if they co-occur in a scene. The $\overline{5}$-DLC assigns the highest centrality value to the protagonist, Jean Valjean, while other major characters such as Javert and Gavroche are similarly ranked highly.
    \item {\bf Network Science}. A collaboration graph from 2006 \cite{newman2006finding} containing 379 researchers in network science. Two authors are linked if they co-authored a paper together. As noted in \cite{newman2006finding}, the top 5 authors with the highest $\overline{5}$-DLC scores are group leaders or senior researchers of groups working in network science. 
\end{itemize}
These observations affirm that DLC may serve as a reasonable centrality measure identifying known important entities in graphs derived from different contexts. While encouraging, such results are not surprising; indeed other centrality measures, such as closeness, betweenness, PageRank, and Katz, would also likely rank some of the same vertices highly. 
 
In order to more comprehensively compare our method to existing centrality methods, next we compare not only the top scoring vertices, but how the entire rankings compare to those given by these other well-known methods. 
Furthermore, we apply our measures under a variety of different parameter settings, considering both DLC and nDLC, and using the top $k$ smallest or largest eigenvalues, for all $k$. 
For each such case, we quantify the similarity in ranking using a well-known nonparametric statistical measure, Spearman's rank correlation coefficient. 
Spearman's $\rho$ ranges from -1 (if the ordinal rankings are reverses of each other) to 1 (if the ordinal rankings are identical). 
Figure \ref{fig:spearman} presents the results for the Network Science collaboration graph, the largest of the 3 data sets. 

Examining the correlation coefficients for closeness and betweenness centrality, neither DLC nor nDLC exhibit strong rank correlations with these measures, regardless of whether the top or bottom $k$ eigenvalues are utilized. On the other hand, PageRank and Katz centrality show very strong correlations with $\overline{k}$-DLC, which tend to increase as more of the top eigenvalues are considered. In contrast, the $\overline{k}$-nDLC rankings consistently show moderate anti-correlations across all four centrality measures.  Across all 4 centrality measures, we observe a rapid convergence between the $k$ and $\bar{k}$ versions of $\DLC$ and $\nDLC$ as $k$ approaches the number of vertices.  While the agreement of these measures is inevitable (as the set of eigenvalues are eventually the same) it is notable this convergence happens over an abbreviated range of choices for $k$ (especially for $\nDLC$).\footnote{Although it is clear that when the entire spectrum is being consider the exact same ordering will be recovered, we note that equal values of the Spearman's $\rho$ (such as between $k$-$\DLC$ and $\overline{k}$-$\DLC$ for PageRank) does not imply that the orderings induced by those measures are the same.}  We suspect this is a general phenomenon  reflective of the tight correlation of small eigenvalues of $L$ and $\mathcal{L}$ with the structure of $G$.  However, as we will see in the next section, there is still significant information that can be recovered by considering both $k$ and $\overline{k}$ variants. 
Overall, the range of differences observed here underscores the flexibility afforded by the parameter settings and choice of matrix in applying our measures, and confirms DLC and nDLC rankings may be distinct from those given by other prominent centrality measures.

\begin{figure}
    \centering
    \includegraphics[scale=0.4]{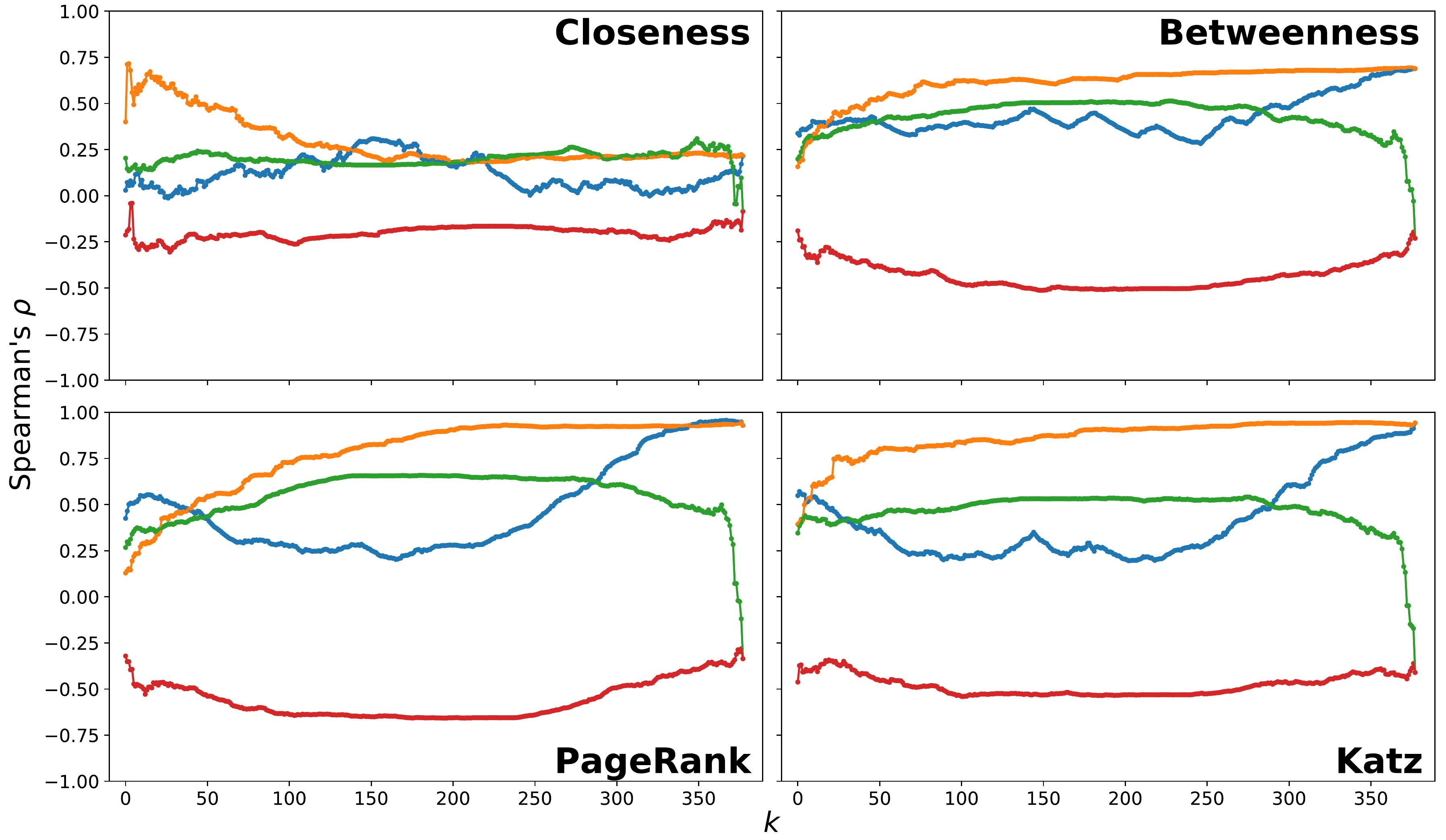} \\
    \hspace{10mm}\includegraphics[scale=0.4]{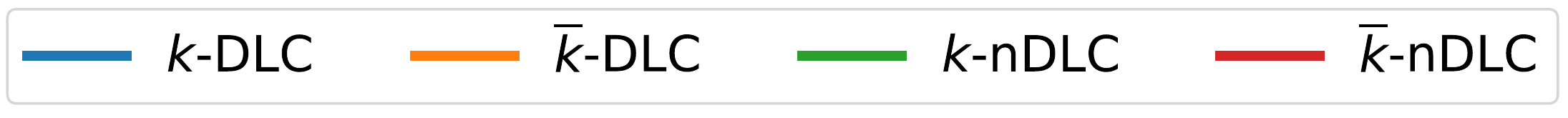}
    \caption{Correlations between rankings of authors in the Network Science data set derived from our Directional Laplacian and Normalized Laplacian centrality under different parameter settings, compared with rankings derived from four other well-known centrality measures. }
    \label{fig:spearman}
\end{figure}

\subsection{Enron E-mail Data}\label{SS:enron}
\afterpage{
\begin{landscape}
\begin{figure}
\centering
\includegraphics[trim = 90 130 5 85, clip, width = \linewidth]{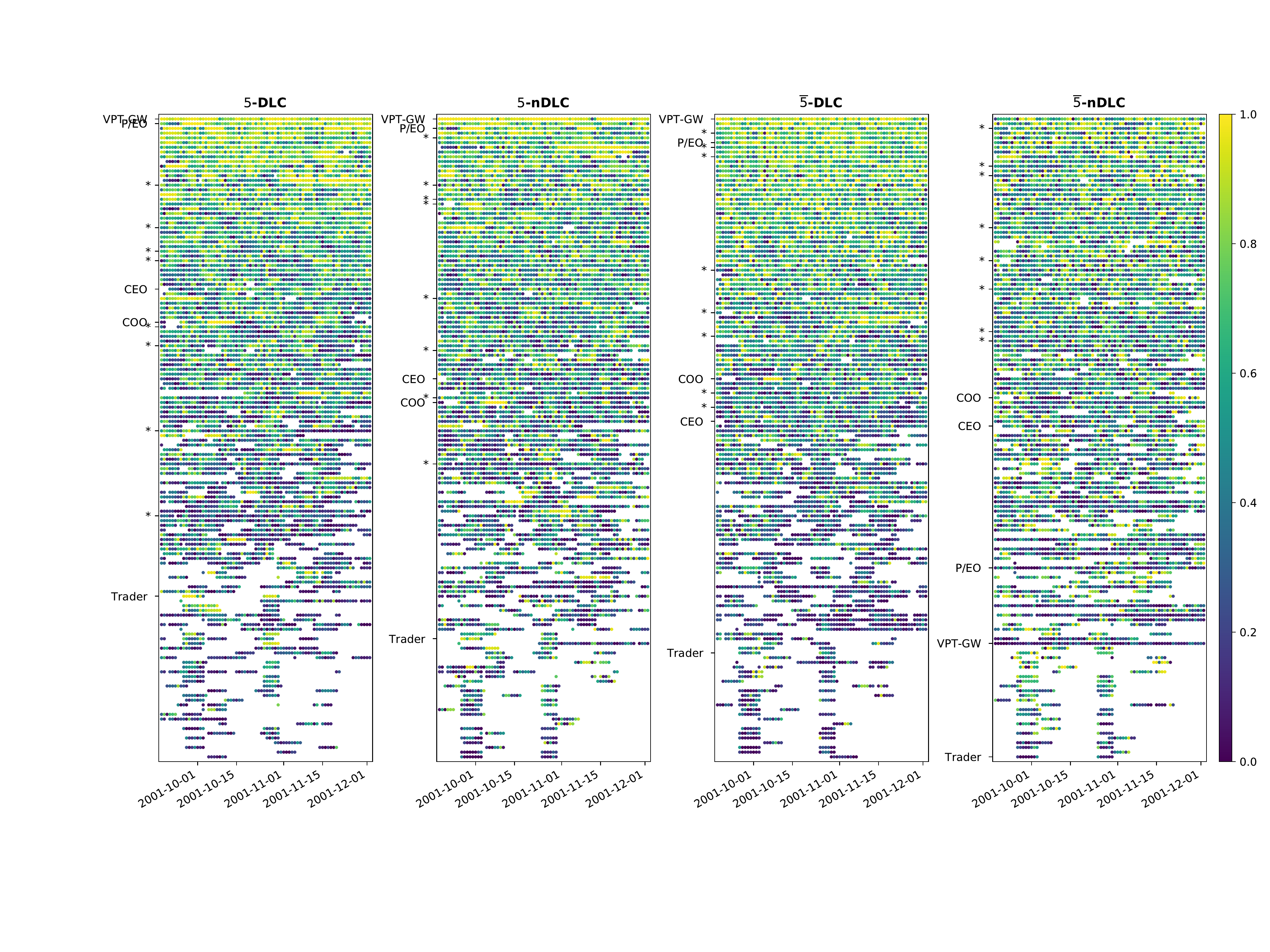}
\caption{Percentiles of $5$-$\DLC$, $5$-$\nDLC$, $\overline{5}$-$\DLC$, and $\overline{5}$-$\nDLC$ of the Enron email day from September $12^{\textrm{th}}$ to December $2^{\textrm{th}}$ (the day before significant layoffs).  Each email graph is formed from the emails for the previous seven days.  Each plot is sorted in order of increasing sum the percentiles over all graphs.  Important/interesting employees in the data set are identified by job title, Chief Executive Officer (CEO), Chief Operating Officer (COO), President and Executive Officer (P/EO), Vice President of Trading -- Gas West (VPT-GW).  The employees labelled with ``*'' are associated with legal oversight of Enron's activities (i.e., those with job titles such as General Council, Employment Lawyer, Legal Specialist, etc.).  The employee identified as Trader is one of several traders in the data set.} \label{F:enron}
\end{figure}
\end{landscape}}
The Enron email data set is a well studied data set~\cite{Priebe2005,Chapanond2005,Leskovec2010,Diesner2005,Bader2008,KlimtYang2004,Shetty:Enron} resulting from the public release of emails of approximately 150 Enron employees gathered by a Federal Energy Regulation Commission (FERC) investigation into the manipulation of energy market prices in the western region (in particular, California) by Enron.  
Subsequently, these emails were deemed a public record and were made available to the public in various cleaned formats\footnote{The analysis in this work is based on the version available at \url{https://s3.amazonaws.com/enron-shetty-adibi/enron-mysqldump.tar.gz}}. 
We restrict our attention to the approximately 13,000 emails in the data set with a time stamp between September $12^{\textrm{th}}$, 2001 and December $2^{\textrm{nd}}$, 2001 (the day before Enron laid-off approximately 4,000 employees). 
For each 7 day window within this time frame, we build the contact graph where two employees are connected by an edge if there is any email communication between them in that window.  
From the graph associated with the window, we calculate the percentile rank according to $5\mbox{-}\DLC$, $5\mbox{-}\nDLC$, $\overline{5}\mbox{-}\DLC$, and $\overline{5}\mbox{-}\nDLC$ for each employee in the giant component.  
The results of these calculations (sorted by the sum of all percentile scores over the entire time frame) are displayed in Figure \ref{F:enron}.  
While an in-depth investigation of these results is beyond the scope of this work, we wish to point out three high level observations resulting from this analysis.

First, we note that for $5$-$\DLC$, $5$-$\nDLC$, and $\overline{5}$-$\DLC$, the most important individual over the period of interest is the Vice President of Trading in charge of the market segment corresponding to gas in the western region (VPT-GW).  
As FERC's investigation was focused on the manipulation of energy prices in California it is unsurprising that the VPT-GW would be a key person of interest and this is reflected in the structure of the email communications gathered by FERC.  
It is worth noting that we are aware of no other method applied to the Enron data set that identifies VPT-GW as the most important individual.

The second observation is that the employees associated with legal issues (denoted by a ``*'' in Figure \ref{F:enron} and including roles such as general council, employment lawyer, legal specialist, etc.) have visually higher overall importance in $\overline{5}$-$\DLC$ and $\overline{5}$-$\nDLC$ than in $5$-$\DLC$ and $5$-$\nDLC$ with an average ranking of the summed percentile importance increasing from $42$ to $32$ for $\DLC$ and from  $34$ to $25$ for $\nDLC$.  
While the precise reason for this shift is not clear, it seems likely that the nature of legal advising within Enron results in a structurally different communication pattern than that associated with trading, to which the bulk of the other employees can be affiliated.

Finally, we note the interesting patterns associated with a single trader in $5$-$\DLC$ (denoted ``Trader'' in Figure~\ref{F:enron}).
While there are several other employees who have a similar pattern of only occasionally communicating with other employees in the data set, this trader is unique in the level of importance during those time periods.  
In particular, this trader's average $5$-$\DLC$ percentile is over $86.1\%$ for the week-long periods when they are present in the giant component.  
While the reason for this behavior is unclear, this trader appears in evidence brought by Snohomish County Public Utility District (SPUD) in their lawsuit against Enron.  
According to contemporaneous reporting on the lawsuit, in the transcripts provided by SPUD this trader arranged for the shut-down of power plants in order to manipulate the California energy market through widespread blackouts/brownouts.

\section{Cyber-Situational Awareness Experimental Setup and Results} \label{sec:exp}

To evaluate the effectiveness of our proposed DLC and nDLC, we analyze their sensitivity to planted anomalies in a graph of network flow. 
Test data sets with real network data and ground truth labeled anomalies can be hard to come by.
There are many causes for this difficulty, for instance, network owners may not want to share technical data regarding their network for fear
that this information will allow adversaries to identify critical resources in a network.  
Additionally, even if network owners are willing to release technical information regarding their network, in many cases even labelling the ground truth is difficult.
However, recently Los Alamos National Laboratory (LANL) has shared three large de-identified data sets\footnote{\url{https://csr.lanl.gov/data/}} from their internal network capture~\cite{Kent:LANL_AuthenticationData,Kent:LANL_RedTeamData,Kent:LANL_RedTeam,LANL:HostData}.
One set in particular, {\it Comprehensive, Multi-Source Cyber-Security Events} (CMCE), contains authentication logs, process logs, network flow, domain name server lookups, and labeled red team authentication events over the course of 58 consecutive days.
Figure \ref{F:density} depicts the number of flow records in trailing 1-minute windows from the LANL CMCE dataset~\cite{Kent:LANL_RedTeamData} as well as timing of the successful red team authentications.
\begin{figure}[b]
    \centering
  \includegraphics[trim = 80 0 100 35, clip, width = .9\linewidth]{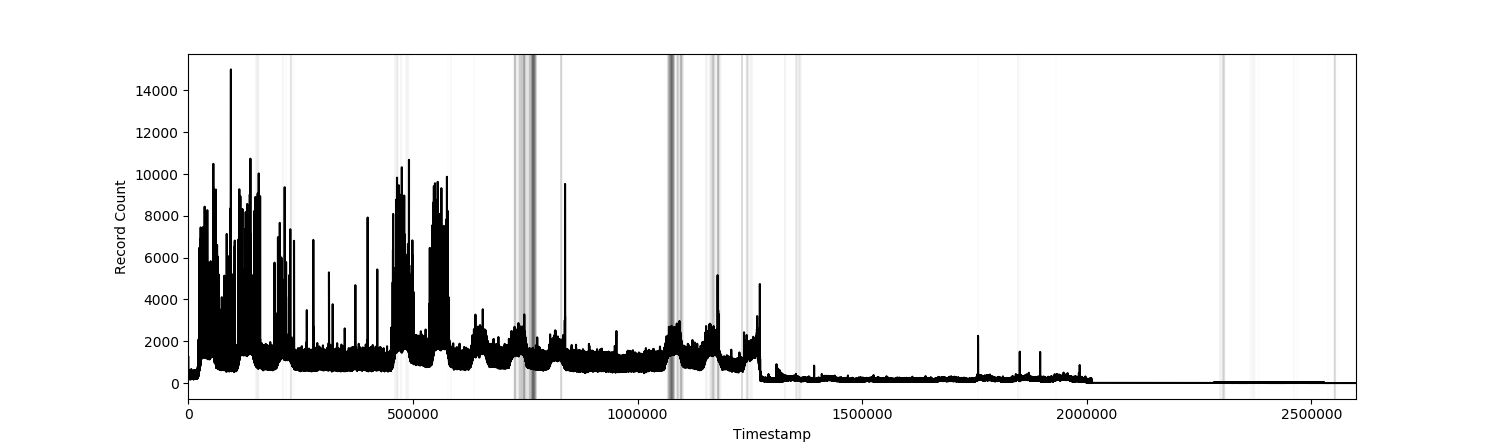}
  \caption{Flow record density of trailing 1 minute window for the LANL data set.  Successful red-team authentications are denoted by vertical gray lines.}\label{F:density}
\end{figure}
It is worth mentioning that the red team events do not represent all of the actions of the red team, but rather contain only the successful authentications from the red team.
This data set will be the source of our baseline graph for two of three experiments.
While we recognize that no two networks are alike, we note that in our experience with diverse operational data sets there are many similarities within network traffic data across networks. 
The resulting graphs tend to have roughly power law degree distribution with few hubs (high degree nodes) and many leaves or short paths.
The ``dandelion'', a high degree node connected to many degree 1 nodes, is a common motif found in these graphs.
We believe that the LANL data is a good surrogate for a generic network as we find these structures to be prevalent in this data set. 
But we also recognize the need for future work to consider data from additional networks to feed our baseline.

In order to have a baseline graph that is not likely to be influenced by red team activities we chose an arbitrary 1-minute time window (timestamps $1065740$ to $1065800$) from the CMCE network flow that corresponds to a typically sized graph and is well-separated from red team events (nearly 3 days after the last successful red team authentications and 10 minutes before the next red team authentication). 
To avoid minor technical issues with the spectra of disconnected graphs the experiments will only use the giant component of this flow graph, which contains 94\% of the vertices. 
This largest component has 2,005 vertices, 2,450 edges, and has a maximum degree of 605.  
Figure \ref{fig:LANL_1min} shows our chosen graph, which was also used in Figures \ref{fig:centEx} and \ref{fig:LC_way} to illustrate differences in centrality scores.
The degree distribution of this graph is roughly power-law with exponent approximately equal to $2.18$, see Figure \ref{F:degreedist}. \begin{figure}[b]
\centering
  \subfloat[Graph visualization\label{fig:LANL_1min}]{\includegraphics[width = .4\linewidth]{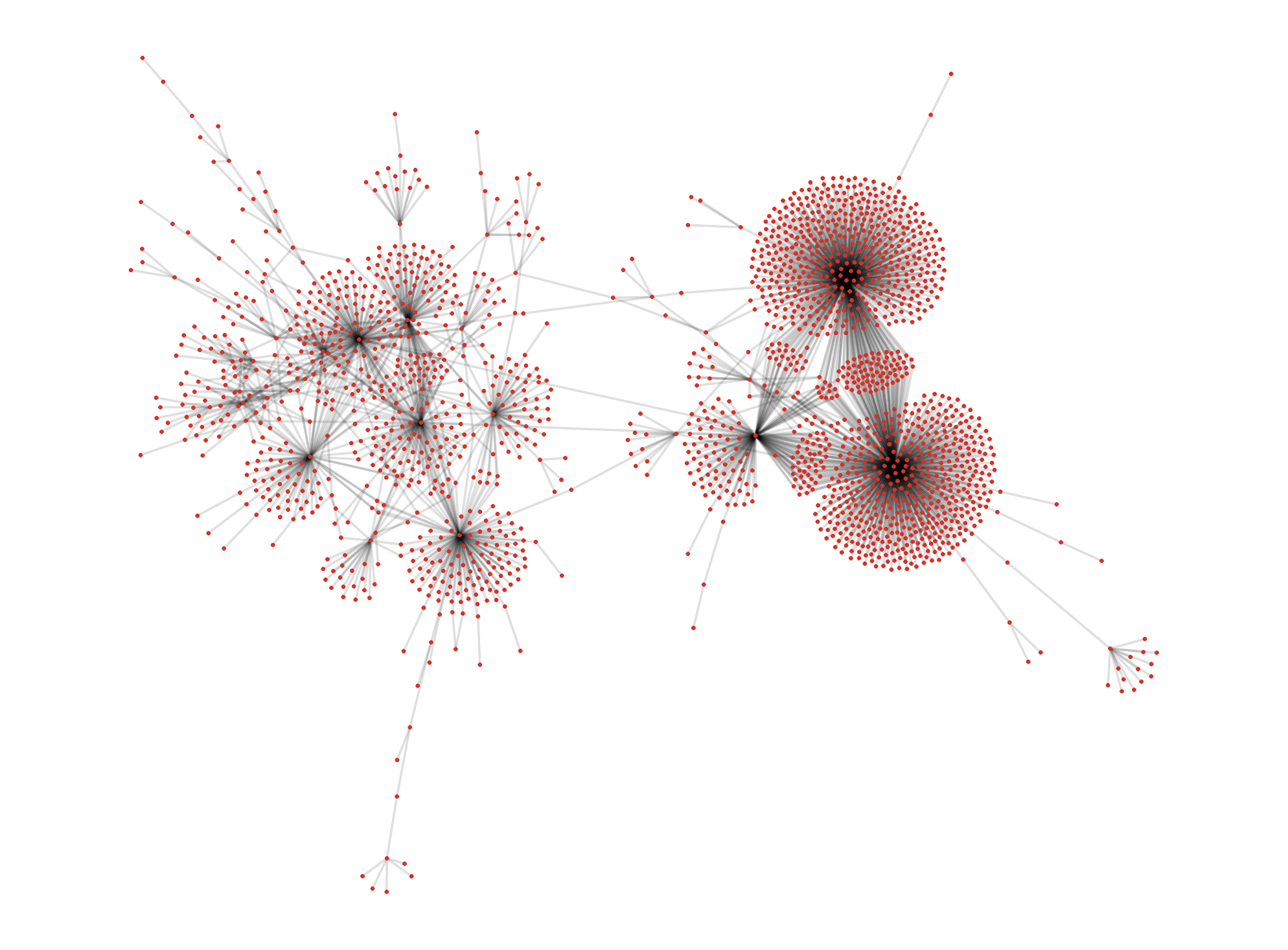}}
  \qquad \qquad
  \subfloat[Degree distribution \label{F:degreedist}]{  \includegraphics[trim = 15 0 40 35,clip, width = .4\linewidth]{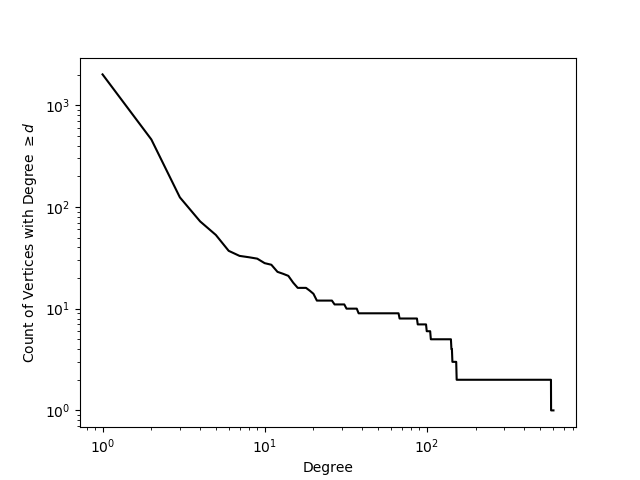}}\\
\caption{Graph visualization of network flow over 1 minute from CMCE data (left) and the degree distribution of that graph (right)}
\end{figure}
Into this specific snapshot from the LANL data we will inject two types of anomalies, represented as extremal subgraphs depicted in Figure \ref{fig:star_clique}, the star and clique.
The specific method of injection varies by experiment and so those details are provided in Sections \ref{subsect:ordered}, \ref{subsect:random}, and \ref{subsect:time}.
These two subgraphs represent building blocks of adversarial behavior: network reconnaissance or data gathering (clique) and lateral movement (star). 
By studying the injection of these two patterns separately we can better understand how their presence will perturb importance values.
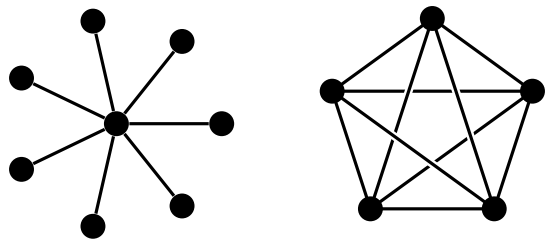
\begin{figure}[b]
\centering
\begin{tikzpicture}[scale=0.7]
\begin{scope}[xshift=6cm]
    \foreach \a in { 18, 90, 162, 234, 306 } {
      \foreach \b in { 18, 90, 162, 234, 306 } {
        \drawPolarLinewithBG{\a:2}{\b:2};
      }
    }
    \foreach \a in {18,90, 162, 234, 306 } {
      \node at (\a:2cm) [circle,fill=black] {};
    }
    \node at (0,-2.5) {};
  \end{scope}
     \stargraph{7}{2}
\end{tikzpicture}
\caption{A star graph (left) and a clique (right).}\label{fig:star_clique}
\end{figure}

{\it Star anomaly: }
In the context of network flow, the injection of a star represents the scenario in which a particular IP address (represented by the central, or {\it root}, vertex $r$) exchanges data with some additional vertices (the leaves) while the rest of the graph remains exactly the same.  
Such a scenario might be malicious or benign, depending on the number of additional IPs contacted and the nature of the data exchange. 
Regardless, this structural change should be reflected as an increase in the importance of $r$ since the increase in number of IPs with which $v$ exchanges data is the sole change in the graph.

{\it Clique anomaly: } The injection of a clique represents the scenario in which a particular set of IP addresses all exchange data with each other while the rest of the graph remains exactly the same. 
Such a scenario could be benign, representing a broadcasting operation in a distributed computing scenario, or could represent malicious behaviour such as a prelude to data exfiltration or command-and-control. 
In either case, the importance of the involved vertices should increase as the inter-set communication represents the sole change in the graph.

Using these anomalies we conduct three injection experiments, described in the next three subsections, with increasing levels of complexity and stringency. 
The first two experiments will be performed on our chosen snapshot of LANL network flow while the third uses a synthetic time-evolving graph.
In Section \ref{subsect:ordered} we consider injection of both patterns among low importance vertices in the baseline graph mentioned above, and in Section \ref{subsect:random} we randomly inject the star anomaly. This simulates what could happen if the only temporal change in a graph is the addition of one of these subgraphs. In Section \ref{subsect:time} we inject the star anomaly within normal temporal evolution.
It is worth mentioning that conceptually simpler measures (such as degree change or local density) may also be able to detect the root of the star-anomaly or the vertices of clique-anomaly when compared to a sufficiently stable background.  However, in the dynamic environment, such as that discussed in Section \ref{subsect:time}, these measures will have significant difficulty in distinguishing the injection of anomalous behavior amid the natural variation of the underlying graph.

\subsection{Anomaly Injection Among Low-Importance Vertices} \label{subsect:ordered}

In this first experiment we deterministically choose vertices with low importance scores to be involved in the anomaly. 
This simple experiment tests whether, all else equal, the change to lower-importance vertices is registered by the centrality metrics. 
The procedure for injecting both anomalies into $G$ starts by ordering the vertices of $G$ according to importance in $G$, lowest to highest. Then, to inject a star anomaly of size $s$ into $G$:
\renewcommand{\theenumi}{\roman{enumi}}
\begin{enumerate}
\item Select a {\it root} vertex $r$ to be one of the vertices of minimum importance.
\item Select the next smallest $s-1$ vertices in the list of importance to be {\it leaves}.
\item Add, if necessary, an edge between the root vertex and any leaf.
\end{enumerate}
Similarly, the procedure for injecting a clique anomaly of size $s$ into $G$ is:
\begin{enumerate}
\item Select a set $S$ of $s$ smallest importance vertices.
\item For every pair of vertices, $u,v \in S$, add, if necessary, an edge between $u$ and $v$.
\end{enumerate}

Figure \ref{F:injection} displays the percentile of the importance values of the vertices involved in the injection as more vertices are added to the anomaly (indicated on the $x$ axis).
For example, in Figure \ref{F:SIL} at $x=100$ the corresponding $y$ value for the solid line is the percentile of the $\overline{5}$-DLC for the root vertex after 100 leaf vertices have been attached to it.
The $y$ value for the dashed line is the percentile of the average $\overline{5}$-DLC over all of the 100 leaf vertices.
For both the leaf vertices and the vertices involved in the clique, the centrality scores are averaged before evaluating the centrality percentile.
We use the percentile rather than the raw centrality score in order to put centralities of all graphs on the same scale. 
Each $x$ value corresponds to a different graph (the baseline graph with a differently-sized anomaly injected) and thus the importance values may not be comparable across all $x$ values. For this reason we will continue to use the percentile in the remaining two experiments.
  
  \begin{figure}[h]
  \centering
  \subfloat[Star Injection for $\overline{5}$-DLC\label{F:SIL}]{\includegraphics[width = .5\linewidth]{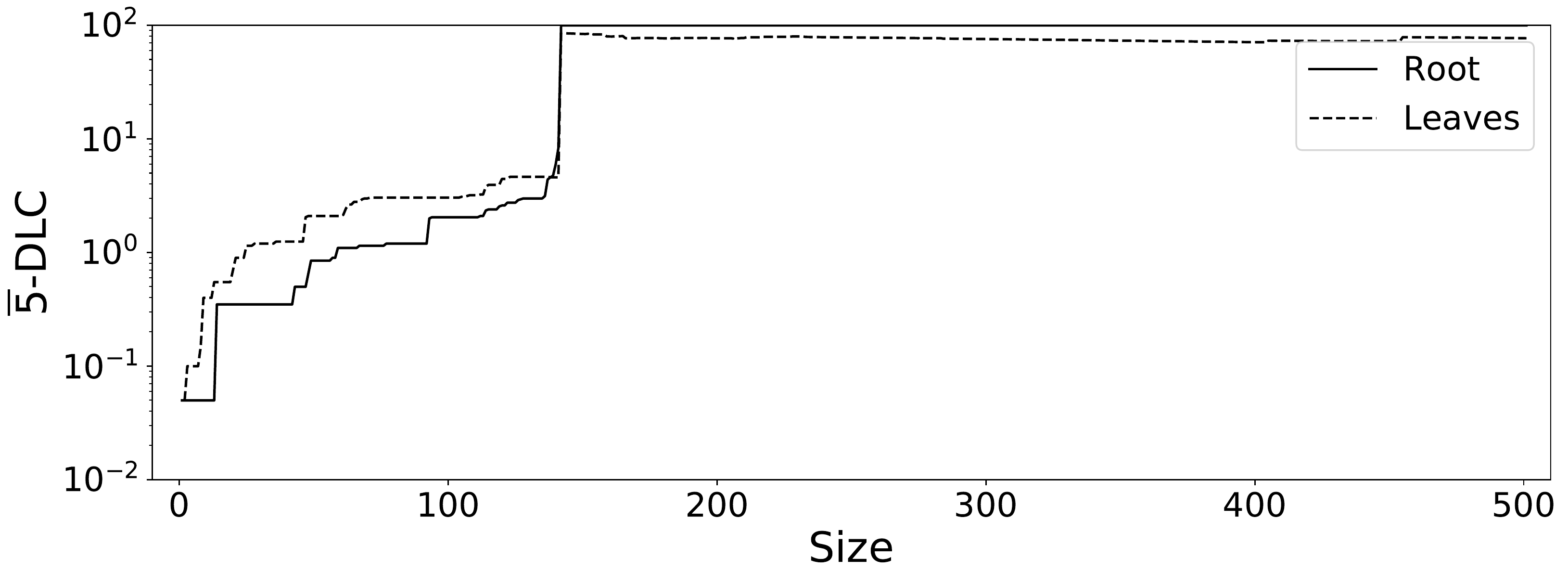}} 
  \subfloat[Star Injection for $5$-nDLC \label{F:SIRW}]{\includegraphics[width = .5\linewidth]{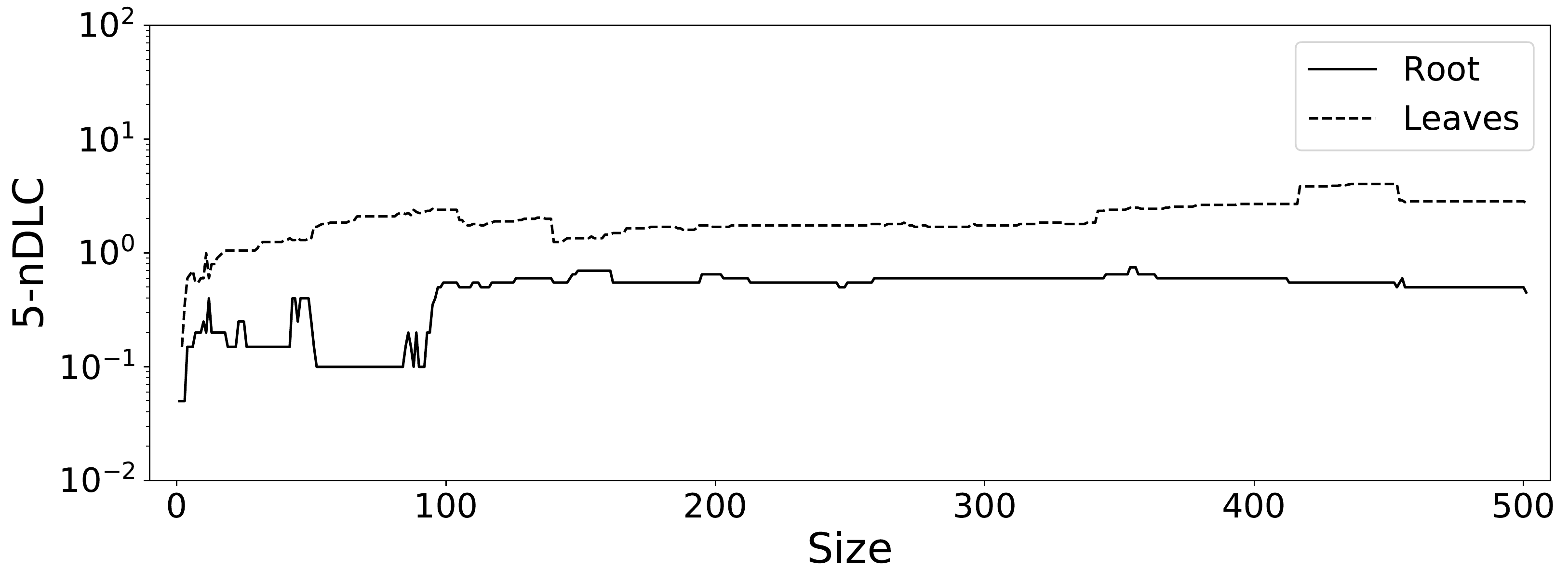} }\\
\subfloat[Clique Injection for $\overline{5}$-DLC \label{F:CIL}]{\includegraphics[width = .5\linewidth]{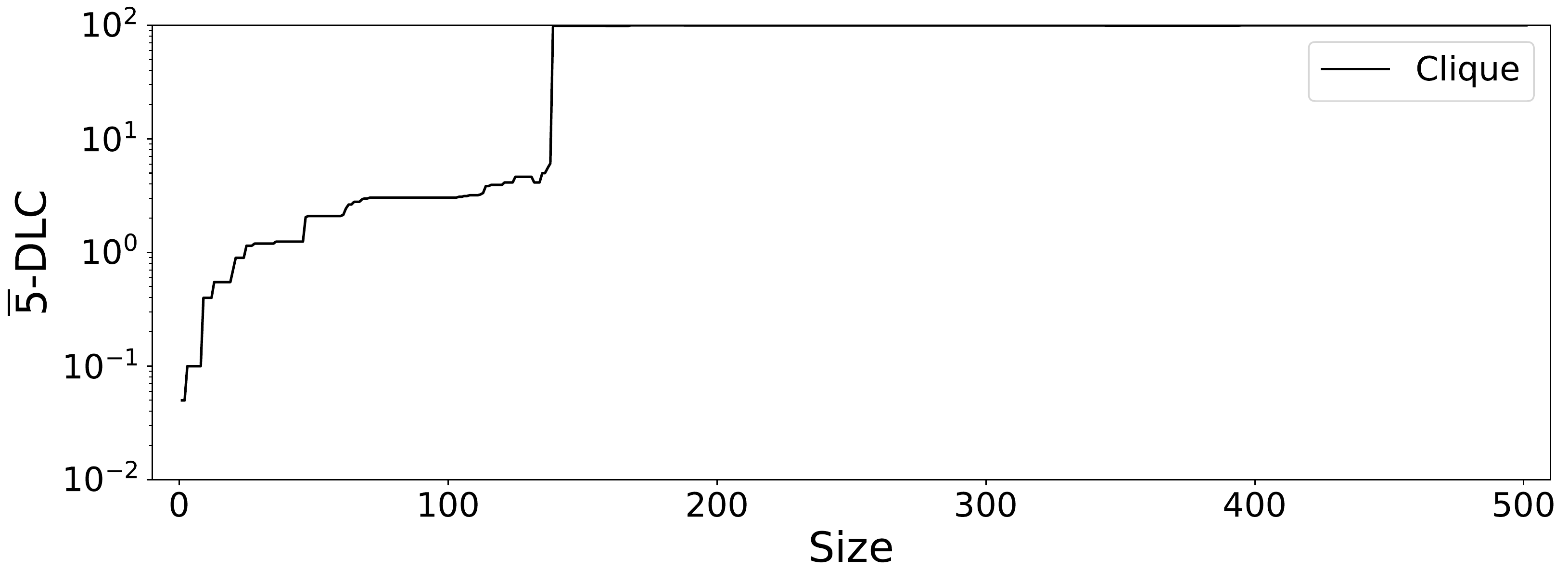}}
  \subfloat[Clique Injection for $5$-nDLC \label{F:CIRW}]{\includegraphics[width = .5\linewidth]{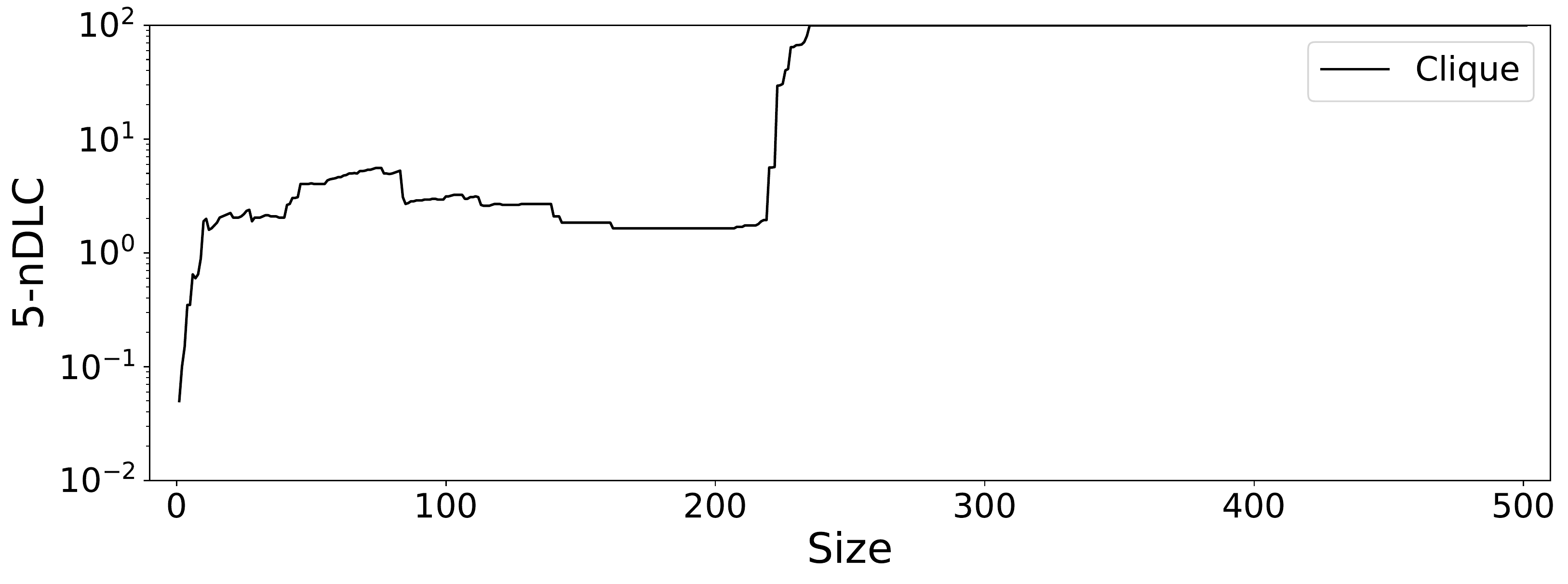} }\\
 \caption{Importance percentile for vertices in injected anomalies as a function of anomaly size.}\label{F:injection}
  \end{figure}

The most striking result of this experiment is that when the injection size is larger than approximately $7.5\%$ of the graph for $\overline{5}$-DLC (Figures \ref{F:SIL} and \ref{F:CIL}) the {\it average} importance of the participating vertices is at or above the $99^{\textrm{th}}$ percentile, representing an increase in importance percentile of more than 3 orders of magnitude.  
A similar increase occurs for the $5$-nDLC under clique injection when the anomaly size is approximately $12.5\%$ of the graph. 
Further, the results of these experiments also clearly indicate that both importance measures can be sensitive to even fairly small injections.  
For instance, when the injection size is approximately $1\%$ of the size of the graph, both importance measures register an order of magnitude increase for the star injection and an approximately 2 orders of magnitude increase for the clique injection.

In the negative direction, it is clear by examining Figure \ref{F:SIRW} that beyond an initial increase in relative importance, the nDLC measure is relatively insensitive to the size of the injected star graph. 
However, this is consistent with the known behavior of the normalized Laplacian.  
Specifically, by applying the theory of quasi-randomness, it can be shown that the spectrum of the normalized Laplacian is sensitive to presence of small cycles in the graph~\cite{Chung:QuasirandomDegSeq,Chung:QuasirandomSparse,Chung:Quasirandom,Chung:QuasirandomPNAS}.  
As a star graph contains no cycles, any small cycles that emerge necessarily have a significant contribution from the underlying graph.  
Further, as the example graph has relatively few edges, it is natural to expect that, even for quite large injections of a star, there will be relatively few short cycles formed in the graph.
  
\subsection{Anomaly Injection Among Random Vertices}\label{subsect:random}

In this next experiment we randomly choose vertices to be involved in the anomaly and consider the results over many trials. 
In this sense, this experiment reflects average case behavior in a static graph, and accounts for the effects of vertex choice in the prior experiment. 
{\it A priori}, there is no reason to believe that the importance measures are independent of the precise set of vertices involved in anomalous behavior.  
In particular, it is possible that the choice of vertex order in the experiments depicted in Figure \ref{F:injection} is a nontrivial factor in the results.  
Thus, to test the sensitivity of these results an additional series of experiments injecting a star anomaly is performed. 
Namely, as before, the  LANL flow graph depicted in Figure \ref{fig:LANL_1min} is utilized as a test case in which structural changes are planted, and the subsequent changes in importance scores are measured over many trials.
The methodology for a single trial is as follows:
\begin{enumerate}
\item Randomly select a vertex $v$.
\item Randomly choose a vertex subset $S$ with $k\%$ of the graph's vertices.
\item For any $s \in S$ that is not already connected to $v$, connect $v$ to $s$.
\item Compute importance score and score percentile of $v$ before and after these additional edge insertions.
\end{enumerate}

By measuring the change in importance score of $v$ for different levels of $k$, this experiment tests the importance measure in two regards: (1) whether the importance measure can {\it detect} the aforementioned structural change (as indicated by a increase in the importance score); and (2) whether the importance measure is {\it sensitive} to the intensity of these changes (as indicated by larger increases in the importance score for larger values of $k$).  In order to reduce the overall computational load and limit the effects of averaging over a large set of vertices, this experiment is restricted to studying the importance of the root vertices when a star is injected.

\begin{table}[h]
\centering
\begin{tabular}{rl||cc|lc|lc|lc|lc}
\toprule
&& \multicolumn{2}{c}{\bf 0.1\%} & \multicolumn{2}{c}{\bf 0.5\%} & \multicolumn{2}{c}{\bf 1.0\%} & \multicolumn{2}{c}{\bf 5.0\%} & \multicolumn{2}{c}{\bf 10.0\%}\\

& & \multicolumn{2}{c}{\footnotesize(2 edges)} & \multicolumn{2}{c}{\footnotesize(10 edges)} & \multicolumn{2}{c}{\footnotesize(20 edges)} & \multicolumn{2}{c}{\footnotesize(100 edges)} & \multicolumn{2}{c}{\footnotesize(201 edges)}\\
& & \multicolumn{1}{c}{\it \footnotesize Score} & \multicolumn{1}{c}{\it  \footnotesize PCTL}  & \multicolumn{1}{c}{\it \footnotesize Score} & \multicolumn{1}{c}{\it \footnotesize PCTL}  & \multicolumn{1}{c}{\it \footnotesize Score} & \multicolumn{1}{c}{\it \footnotesize PCTL} & \multicolumn{1}{c}{\it \footnotesize Score} & \multicolumn{1}{c}{\it \footnotesize PCTL}  & \multicolumn{1}{c}{\it \footnotesize Score} & \multicolumn{1}{c}{\it \footnotesize PCTL}      \\ \midrule

\multirow{3}{*}{\rotatebox[origin=c]{90}{\footnotesize $\aunderbrace[L1U1R]{\overline{5}\mbox{-$\DLC$}}$}}& Before &\texttt{2.30} & {44\%} &  \texttt{1.08} & {45\%}   & \texttt{1.11}  & {45\%}  & \texttt{0.83}  & {43\%}  & \texttt{2.04}  & {46\%}   \\
 & After &\texttt{2.31} & {55\%} & \texttt{1.11}  & {64\%}  & \texttt{1.19}  & {67\%}  & \texttt{1.47}  & {74\%}  & \texttt{205.68}  & {99\%}   \\
& Change &\texttt{0.01} & {11\%} & \texttt{0.03}  & {20\%}  & \texttt{0.08} & {22\%}   & \texttt{0.64}  & {31\%}  & \texttt{203.63}  & {53\%}   \\
\rule{0pt}{3ex}
\multirow{3}{*}{\rotatebox[origin=c]{90}{\footnotesize $\aunderbrace[L1U1R]{5\mbox{-$\nDLC$}}$}}& Before &\texttt{-4e-5} & {50\%} &  \texttt{-7e-5} & {52\%}   & \texttt{-9e-7}  & {50\%}  & \texttt{5e-5}  & {49\%}  & \texttt{-9e-6}  & {51\%}   \\
 & After &\texttt{-8e-4} & {1\%} & \texttt{-2e-3}  & {1\%}  & \texttt{-4e-3}  & {1\%}  & \texttt{-2e-2}  & {1\%}  & \texttt{-3e-2}  & {1\%}   \\
& Change &\texttt{-7e-4} & {-49\%} & \texttt{-2e-3}  & {-51\%}  & \texttt{-4e-3} & {-49\%}   & \texttt{-2e-2}  & {-48\%}  & \texttt{-3e-2}  & {-50\%}   \\
\bottomrule
\end{tabular}
\caption{Importance scores for a vertex before and after star injection according to the proportion of other vertices to which the chosen root vertex connects. The results are averaged over 500 trials.   } \label{tab:injScan}
\end{table}

The results of the experiment are presented in Table \ref{tab:injScan}. 
The experiment was run for $k=$0.1, 0.5, 1, 5, and 10\%, which corresponds to star anomalies with 2, 10, 20, 100, and 201 edges incident to the root $v$. 
The table presents the mean scores and mean percentile before and after the star is inserted, along with the mean change, over 500 trials. 
Turning attention first to the $\overline{5}$-DLC scores, we observe a positive net change.
Even for the smallest value of $k$ in which only two additional edges are added, the $\overline{5}$-DLC scores are responsive, registering an average score percentile increase of 11\%. 
Additionally, we note this change in percentile increases monotonically for each value of $k$, with the largest value of $k$ registering an increase of 53\% placing the root vertex in the $99^{\textrm{th}}$ percentile. 
These results suggest $\overline{5}$-DLC scores can both detect the star anomaly at minuscule levels of intensity, and reflect the relative magnitude of the star anomaly.

In contrast, the nDLC scores {\it decrease} in response to the star anomaly. 
Furthermore, these scores decrease more as the star anomaly intensity parameter $k$ increases. 
This suggests nDLC scores can detect, and are sensitive to, the star anomaly, but in the opposite direction of the DLC score values. 
To explain why this apparent ``inverted scale" occurs for nDLC, the aforementioned theory of quasi-random graphs \cite{Chung:QuasirandomDegSeq,Chung:QuasirandomSparse,Chung:Quasirandom,Chung:QuasirandomPNAS} is again relevant. 
By this theory, the prevalence of short cycles greatly impacts the spectrum, and the star injection is likely to increase the number of such cycles incident to the root vertex $v$. 
The net effect of this change on the nDLC score is ultimately negative because the normalized Laplacian matrix shifts and inverts the sign of the spectrum of the matrix $D^{-1/2}AD^{-1/2}$. 
It is worth noting these observations are in agreement with the results presented in Figure \ref{F:SIRW}, in which root vertices exhibited smaller nDLC values than leaf vertices. 
However, more research is necessary to better interpret nDLC scores, and to reconcile these observations with the clique experiment of Figure \ref{F:CIRW} in which the nDLC scores increase. 
This dichotomy suggests that nDLC scores can have seemingly counterintuitive interpretations, while DLC scores afford simpler interpretations in the contexts considered here.

\subsection{Anomaly Injection with Time-Evolution}\label{subsect:time}

The prior two experiments show that on a static graph the inclusion of an anomaly is reflected in a change in DLC or nDLC. 
However, for situational awareness one must compare a current snapshot with an anomaly to a {\it prior} snapshot without an anomaly. 
In this way we wish to see that the DLC or nDLC are also perturbed significantly in the presence of an anomaly injected at a specific time step, aiding an analyst in discovering the anomalous vertex when comparing importance scores across time steps. 
To that end, in our final experiment we assume a star anomaly occurs concurrently with the natural time-evolution of a network.
This tests the robustness of the measures in the presence of natural noise. 
In order to simulate the natural time-evolution of a cyber network, we require a temporal graph model. For this purpose, we utilize an extension of a model developed by Hagberg, Lemons, abd Misra that was designed to capture temporal dynamics observed in cyber data.

\subsubsection{Hagberg-Lemons-Misra Model}\label{sec:HLM}
Temporal variations in network traffic comprise a mix of fairly steady computer-generated traffic and more variable human-generated traffic. 
Overall, in the data we have observed, beyond some temporal variability (e.g., more traffic during the day than overnight) the structure of communications, as reflected by the degree distribution of graphs for small time windows, is fairly stable. 
The temporal graph model developed by Hagberg, Lemons, and Misra \cite{hagberg2016temporal}, which we refer to as the {\it HLM model}, allows the user to specify a degree sequence and builds on the Chung-Lu model \cite{chung2002average,chung2004average} which generates graphs consistent with a given expected degree sequence.
HLM accounts for temporal variability by ``masking'' some vertex pairs at each time step. 
Those edges in the current time step that correspond to vertex pairs not in the masking set are carried to the next time step, while vertex pairs that are in the masking set are included in the next time step based on independent Bernoulli trials, regardless of whether they correspond to edges in the prior time step. 
All time steps of the HLM evolution are equal in distribution, and equal to a Chung-Lu model with the initial degree sequence.
Recent work by the authors of the present paper, together with Jenne
\cite{aksoy2020inprep}, has extended this model to include a temporally varying degree distribution, now the Temporal HLM or {\it \tHLM} model, and shown how to rapidly generate \tHLM\ instances. 
In order to test our DLC and nDLC measures in realistic time-varying data we will use a \tHLM\ instance with parameters measured from the LANL data set.

While the \tHLM\ model can be based on a variety of random graph models, for these experiments we will consider the variant based on the Chung-Lu random graph model.  Formally, this begins with a degree sequence $w = \tup{w_v}_{v \in V}$ as an $n$-dimensional vector, where $n=|V|$ is the total number of vertices present across all graphs. 
The value $w_v$ is the expected degree of vertex $v$ across all time instances. 
The \tHLM\ model additionally requires a vector that captures temporal variation in graph density, $\tau = \tup{\tau_t}_{t=1}^T$. 
This corresponds to the expected degree of vertex $v$ at time step $t$ being $\tau_t \cdot w_v$ when there is no temporal correlation.  
A final necessary parameter, $\alpha \in [0,1]$ controls the aforementioned masking set and represents the amount of temporal correlation.  That is, if $\alpha = 1$ there is no correlation between successive graphs and if $\alpha = 0$, successive graphs are identical.  
One could consider $\alpha$ to be a vector in $[0,1]^{\binom{n}{2}}$, with one $\alpha$ value for each edge, but for simplicity in this experiment we use a single value for all edges. 
Given these parameters, the \tHLM\ model generates a sequence of graphs $G_1, G_2, \ldots, G_T$.
The graph at time $t$, $G_{t+1}$, is formed from $G_t$ using the $\alpha$ parameter to control for how much of $G_t$ is carried forwards. Specifically, a masking set $M_t$ is generated where each pair of vertices, $\{u,v\}$, is in $M_t$ independently with probability $\alpha$. 
Any pair $\{u,v\} \notin M_t$ which is an edge in $G_t$ is carried forward to be an edge in $G_{t+1}$, and any unmasked pair of vertices which is not an edge in $G_t$ will not be in $G_{t+1}$. 
Then, any pair of vertices $\{u,v\} \in M_t$ is present as an edge in $G_{t+1}$ independently with probability\footnote{In order to avoid probabilities greater than 1 we actually use $\min\left\{1, \frac{\tau_{t+1}w_u w_v}{\rho}\right\}$ as the probability here.} $\frac{\tau_{t+1}w_u w_v}{\rho}$, regardless of whether $\{u,v\}$ was an edge in $G_t$, where $\rho = \sum_{v \in V} w_v$. 
Putting this all together we see that the probability of an edge in $G_{t+1}$ is given by
\[ 
P(\{u,v\} \in G_{t+1}) =
\left\{
\begin{array}{ll}
1-\alpha + \alpha\cdot \frac{\tau_{t+1}w_u w_v}{\rho}    & \{u,v\} \in G_t \\
\alpha\cdot \frac{\tau_{t+1}w_u w_v}{\rho}     &  \{u,v\} \notin G_t
\end{array}
\right..
\]
Unlike in the HLM model, which does not allow the degree sequence to be time varying, $G_t$ is not quite equal in distribution to  $\mathcal{G}(\tau_t\cdot w)$, the Chung-Lu model with expected degree sequence given by $\tau_t\cdot w$.
Rather, the edge $\{u,v\}$ is present in $G_t$ with probability
$(1-\alpha)^t \frac{\tau_1 w_u w_v}{\rho} + \sum_{i=1}^t \alpha(1-\alpha)^{t-i} \frac{\tau_i w_u w_v}{\rho}$.

\subsubsection{Experimental Results} 
To determine the parameters of the \tHLM\ model into which we will inject the anomaly, we take a series of 15 snapshots of one-minute traffic intervals in a 5-minute window surrounding our example LANL graph (each of these snapshots is shifted 20 seconds in time).  
The $w$ vector of \tHLM\ is determined by the average degree sequence (rounded up) of the vertices across these snapshots (with vertices which do not appear in a snapshot counting as having degree 0).  
The resulting weighting vector has 3,987 entries, with a maximum value of 617, and a average value of $1.74$.\footnote{It is worth noting that, because of the extreme sparsity and highly skewed nature of the degree sequence, a standard Chung-Lu model with this parameter has maximum expected degree significantly smaller than $617$.  In fact the expected maximum degree is approximately $446$ and the expected average degree is $1.62$.  To be reflective of this difference we will refer to the value of a vertex in the $w$ as its weight, instead of the more typical usage of expected degree.} 
To mimic the natural circadian rhythms expected in a cyber-security system, the temporal weighting sequence, $\tau$, is given by 500 equally spaced evaluations of $\nicefrac{(3-\cos(x))}{2}$ over the range $[0, 4\pi]$. 
The resulting degree distribution and temporal sequence are illustrated in Figure \ref{F:tHLM_W}.

\begin{figure}
    \centering
    \hfill
    \subfloat[$\log$-$\log$ Complimentary Cumulative Weight Sequence]{\includegraphics[width = .4\linewidth]{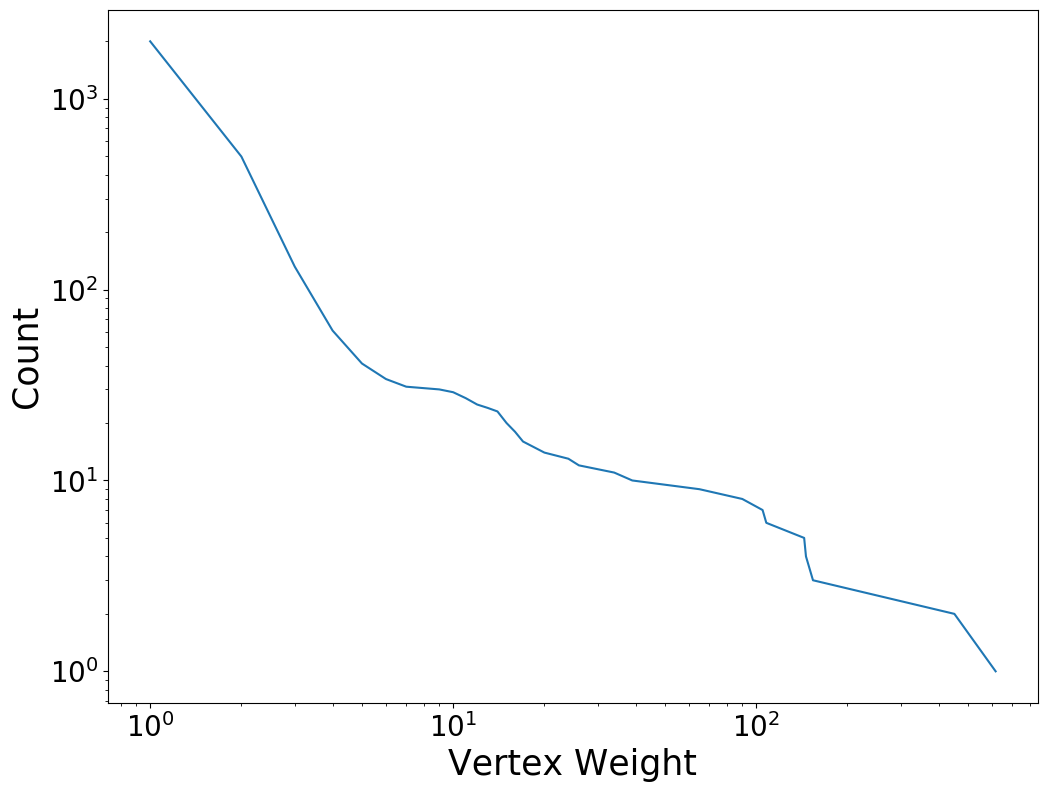}} \hfill
    \subfloat[Temporal Sequence]{\includegraphics[width = .4\linewidth]{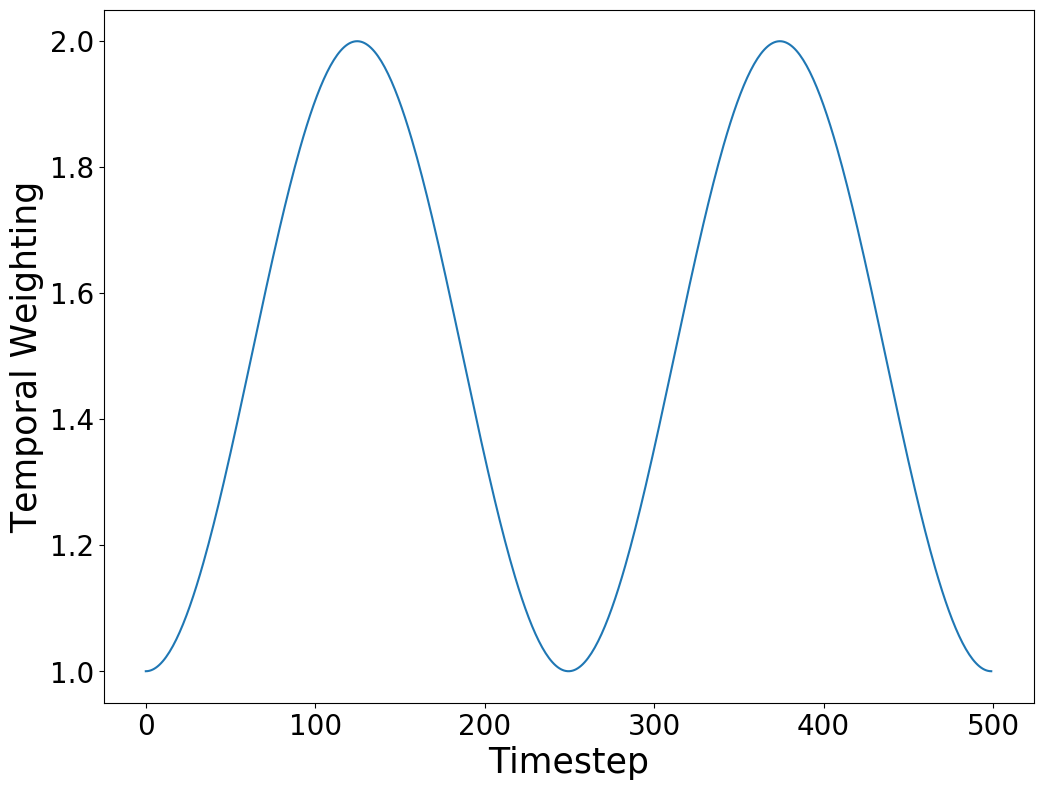}}
    \hfill \phantom{}
    \caption{Parameters for \tHLM\ model}
    \label{F:tHLM_W}
\end{figure}

To simulate a relatively slowly changing environment we set the evolution parameter at $\alpha=0.05$ which means that every edge has the potential to change state approximately 25 times during the course of the evolution.  
The result of the \tHLM\ model with these parameters is a sequence of 500 graphs with significant variation of parameters as shown in Table \ref{T:tHLM}. 
The evolution of these parameters is shown in Figure \ref{F:tHLM}.
\begin{table}[]
    \centering
    \begin{tabular}{l |c c c} 
         & Minimum & Maximum & Mean \\ \hline
         Giant Component Size & 2169& 3473& 2949.3 \\
         Giant Component Edges  &2863 & 6282 & 4669.7 \\
         Maximum Degree & 452 & 896 & 660.9 \\
         Total Edges & 3123 & 6327 & 4795.1 \\
         
    \end{tabular}
    \caption{Statistics for graphs in \tHLM\ sequence.}
    \label{T:tHLM}
\end{table}

\begin{figure}
    \centering
    \hfill
    \subfloat[Giant Component Size]{\includegraphics[width = .3\linewidth]{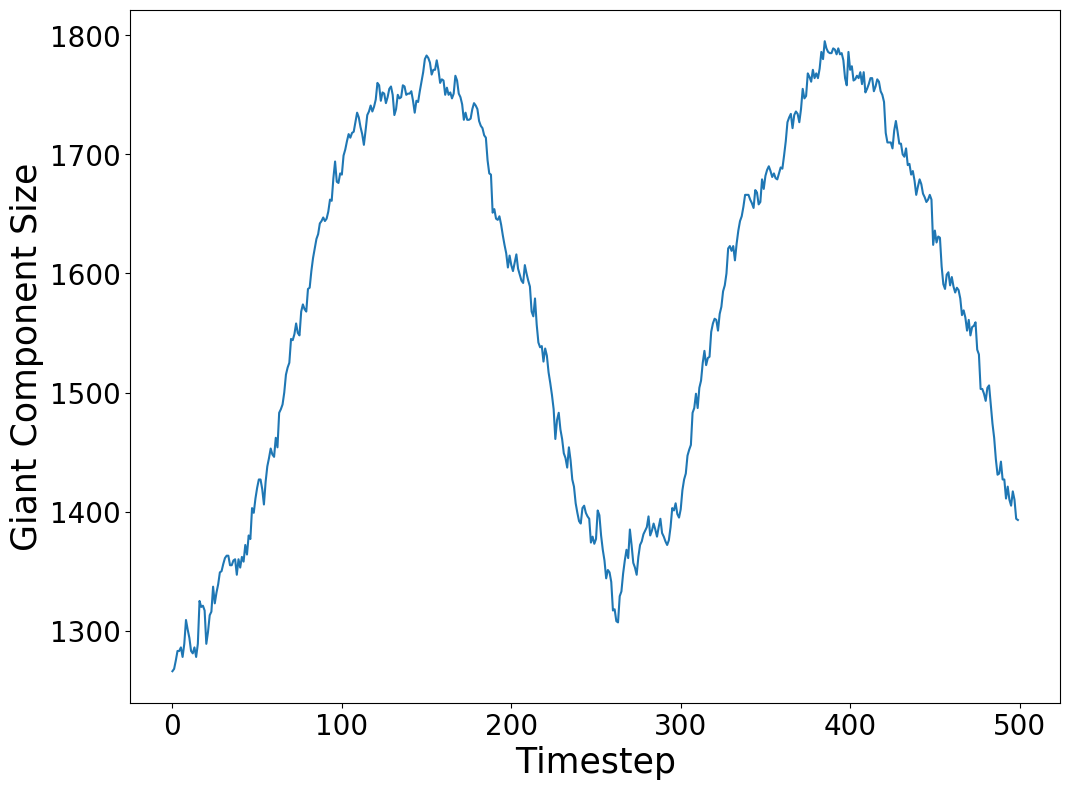}}
    \hfill
    \subfloat[Number of Edges]{\includegraphics[width = .3\linewidth]{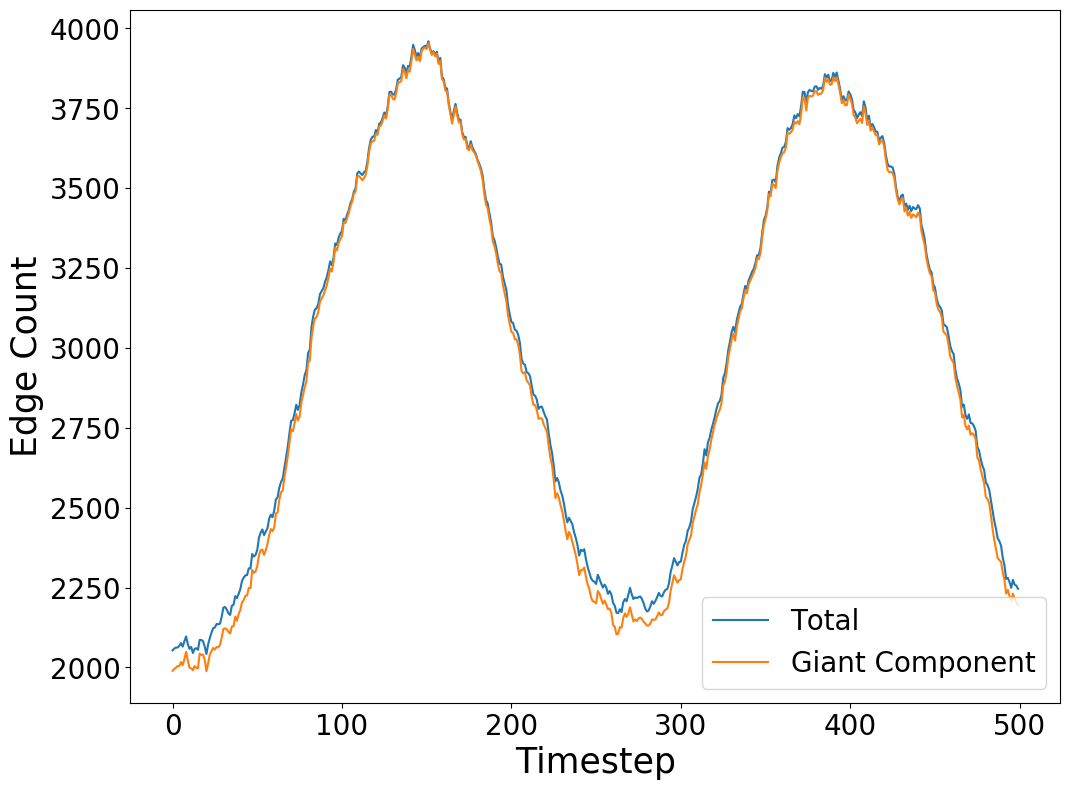}}
    \hfill
    \subfloat[Top 5 Degrees]{\includegraphics[width = .3\linewidth]{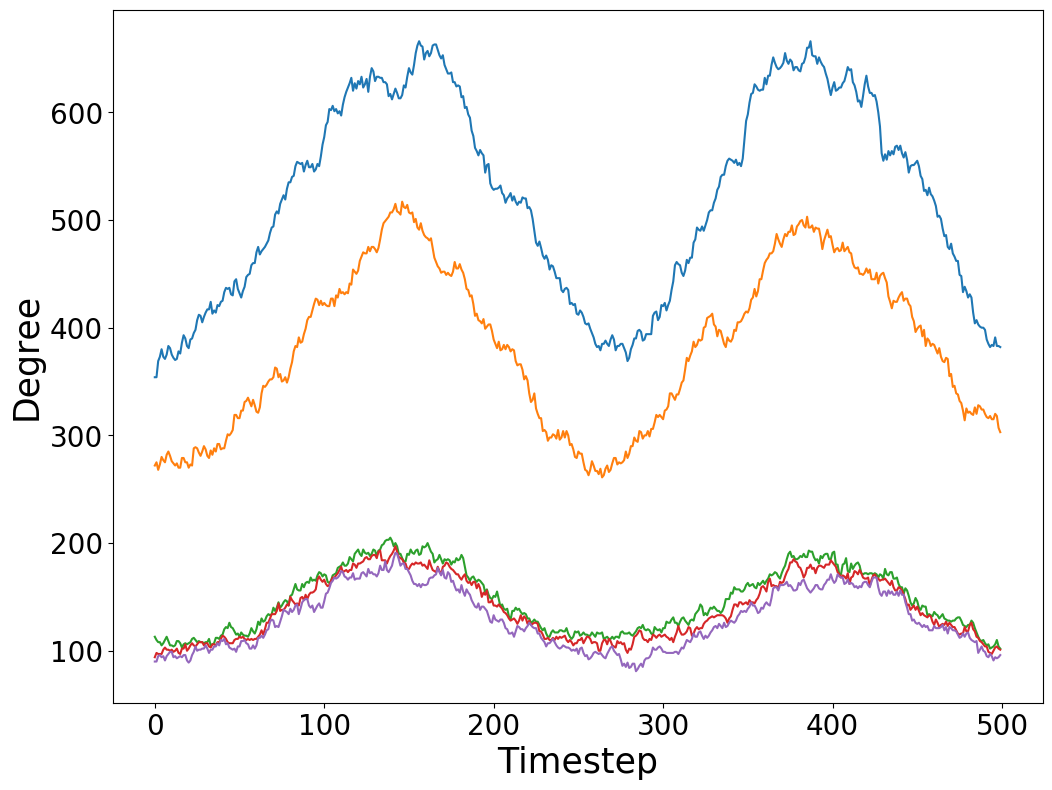}}
    \hfill\phantom{}
    \caption{Summary Statistics for Generated \tHLM\ Model}
    \label{F:tHLM}
\end{figure}

For each time step in the sequence of graphs we will compare the generated graph at time $t$ with the same graph after injecting a single star anomaly with 30 leaves. 
For consistencies sake we will be injecting the anomaly to the same vertices across all time steps.
Before selecting the vertices that will participate in the injected anomaly, we observe that almost all vertices have some time step in the \tHLM\ model where they are in a small component (i.e., not in the giant component).  
As the block structure of the combinatorial and normalized Laplacians respects the component structure of the underlying graph, typical eigenspaces are associated with a single component and are identically zero on all other components.  
The sole exception to this is the case where multiple components have a common eigenvalue, in which case there does exist a basis for the associated eigenspace where all basis vectors are non-zero on all the components with the common eigenvalue.  
Even for these eigenspaces it is typical, for computational reasons, to choose a basis that respects the underlying component structure. 
We will respect this convention in our eigenbasis decomposition.  
As a consequence, when considering $\bar{5}\mbox{-}\DLC$ and $5\mbox{-}\nDLC$, at any particular time step there are at most 5 components whose vertices have non-zero importance.  
In practice, there is typically only a single component whose vertices have non-zero importance, the giant component.  
As a consequence, the injection of any anomaly may trivially change the importance of the vertices involved by adding them to the giant component.  
To avoid this issue, we will choose the anomaly from the set of 66 vertices which are in the giant component for every time step.  
This set is naturally biased to those vertices of higher weight, with an average weight of $36.1$ and so we further refine this set by remove the top $20\%$ of vertices in terms of weight resulting in a set $X$ of candidate vertices for the anomaly with an average weight of $7.1$. 
The root and thirty leaves forming the anomaly, $\mathcal{A}$, are then selected randomly from the remaining 53 vertices.

Now, given the observed graph $G_t$ at some time step $t$, we wish to understand how the behavior of $\bar{5}\mbox{-}\DLC$ and $5\mbox{-}\nDLC$ differs between $G_{t+1}$ and $G_{t+1} + \mathcal{A}.$  
To this end, for each vertex $v$ in $\mathcal{A}$ and for each time step $t$, we identify a set of vertices $S_{v}^{(t)}$ such that for all $s \in S_v$ the percentile ranking of the (normalized) DLC in $G_t$ differs from $v$ by at most $2.5\%$ and $s$ has a non-zero importance in $G_{t+1}$.\footnote{The sets $S_v$ have an average size of 182.4 and 190.4, for $\bar{5}\mbox{-}\DLC$ and $5\mbox{-}\nDLC$, respectively.}  We refer to this set as the \emph{cohort} associated with $v$.
The former condition allows us to aggregate all vertices sufficiently similar to $v$, while the second condition accounts for the fact that the vertices of $\mathcal{A}$ were chosen in such a way that they have non-zero importance for all time steps.  
By considering the percentiles of the (normalized) DLC for vertices in $S_v$, we can recover a "typical" evolution of the importance of a vertex like $v$ after one time step.

\begin{figure}
    \centering
    \hfill
    \subfloat[Kernel Density Estimate for the distribution of the gap between the scores with and without the injected anomaly. \label{SF:gap}]{\includegraphics[width=.45\linewidth]{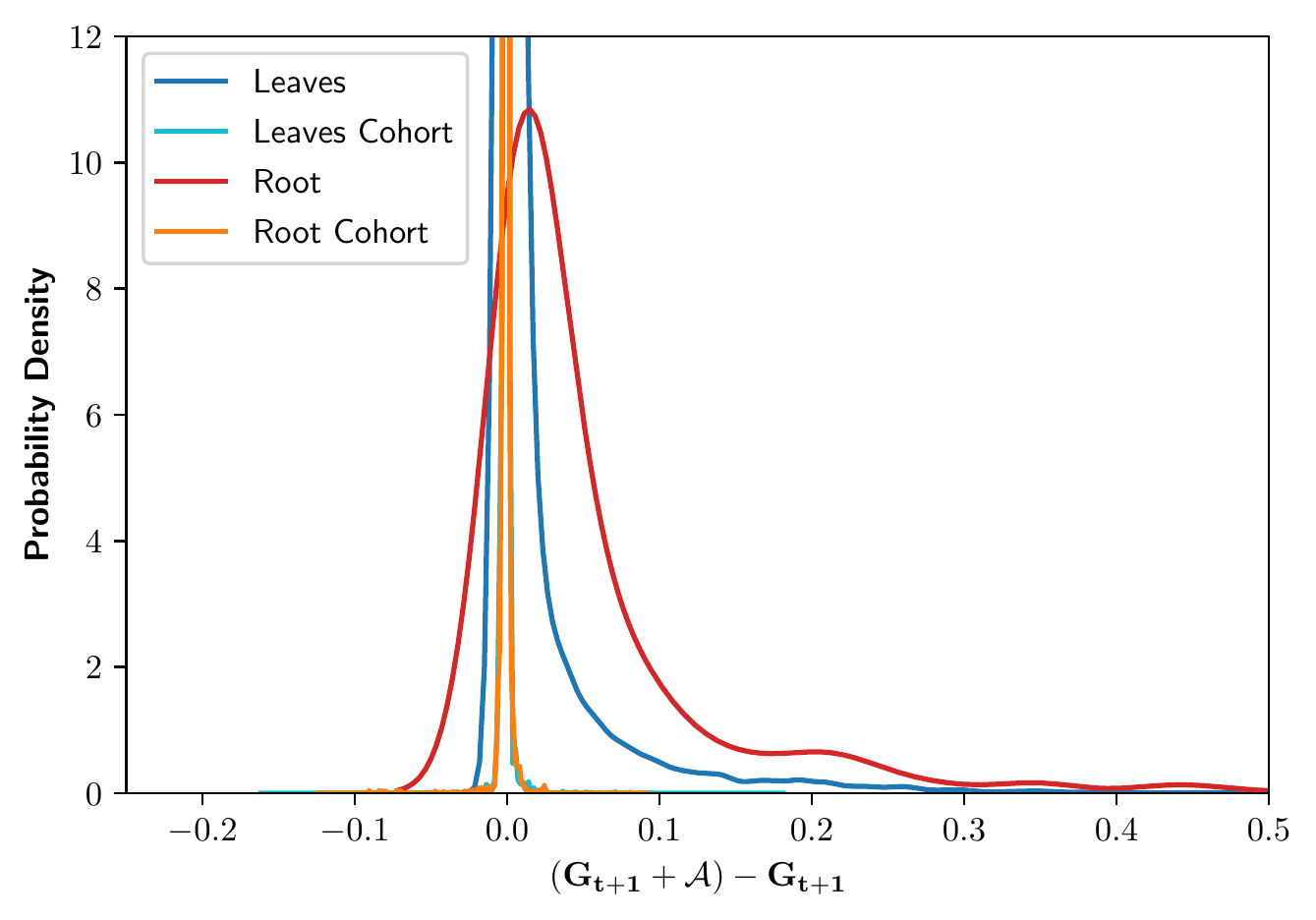}}
    \hfill
    \subfloat[Kernel Density Estimate for the cumulative density function for the gap between the scores with and without the injected anomaly.\label{SF:gap_CDF}]{\includegraphics[width=.45\linewidth]{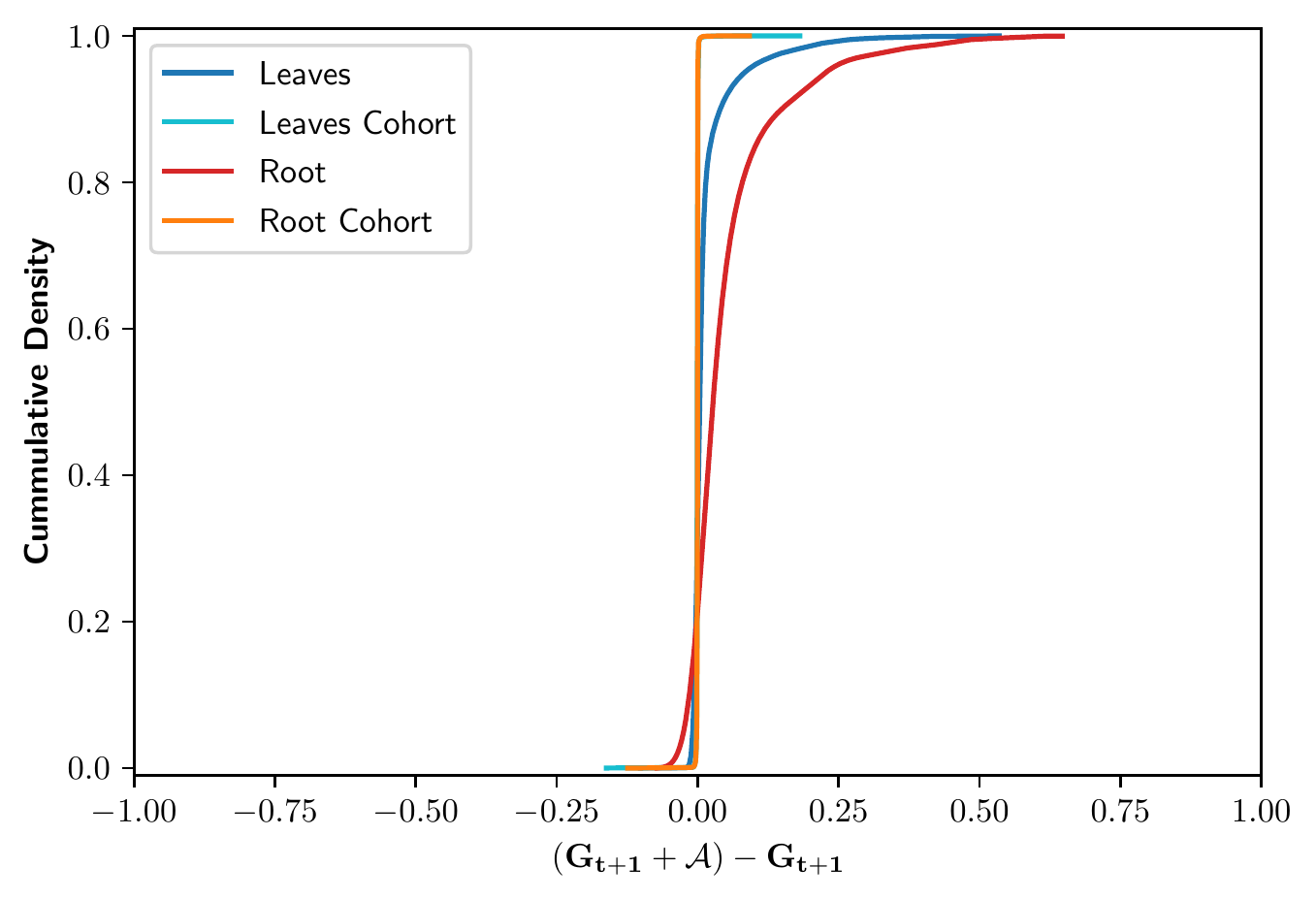}}
    \hfill\phantom{}\\
    \hfill
    \subfloat[Kernel Density Estimate for the distribution of the gap in between the scores with the injected anomaly and the score in previous time step.\label{SF:evo}]{\includegraphics[width=.45\linewidth]{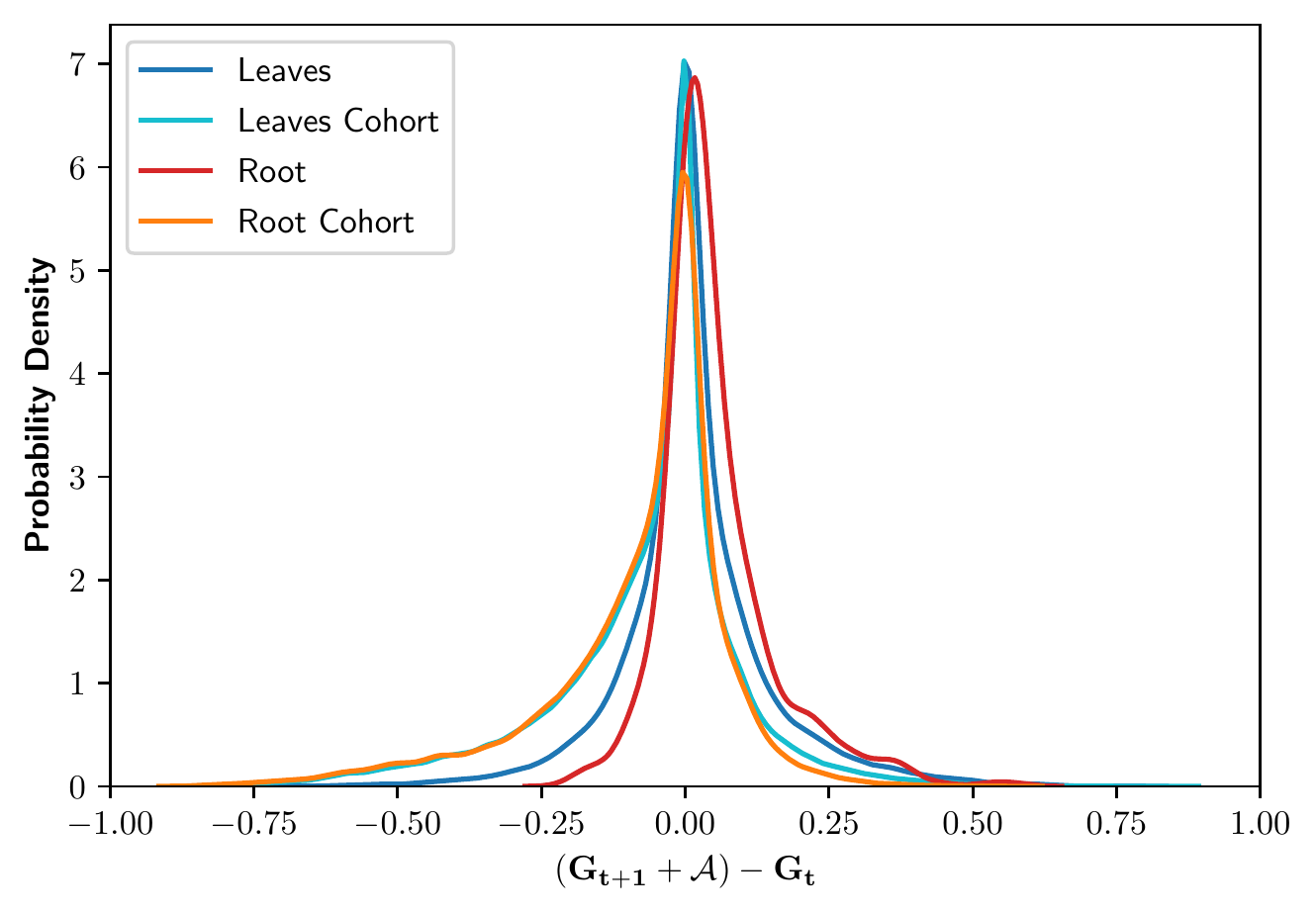}}
    \hfill
    \subfloat[Kernel Density Estimate for the cumulative density function for the gap in between the scores with the injected anomaly and the score in previous time step.\label{SF:evo_CDF}]{\includegraphics[width=.45\linewidth]{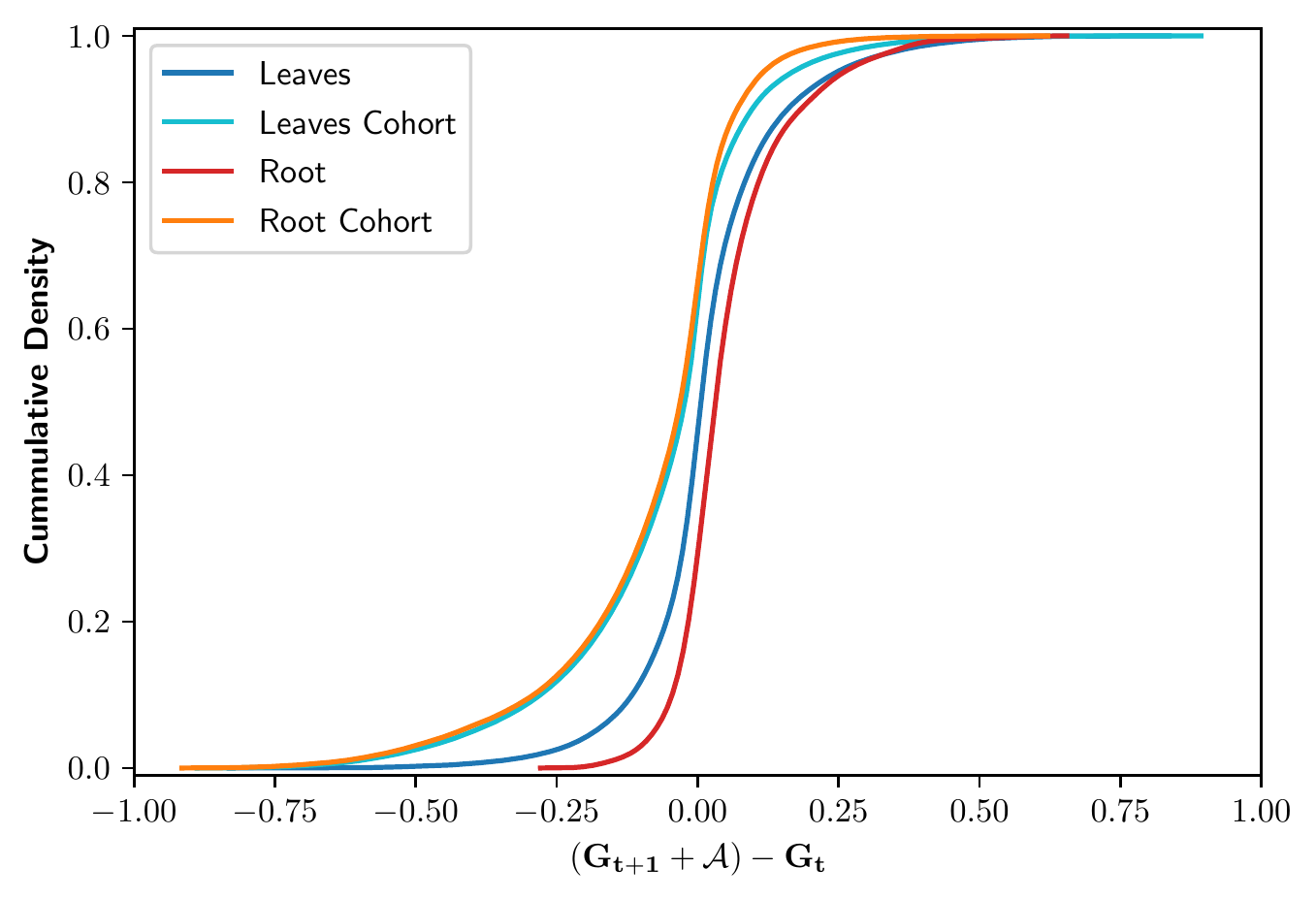}}
    \hfill\phantom{}
    \caption{$\bar{5}\mbox{-}\DLC$}
    \label{F:tHLM_DLC}
\end{figure}

In Figure \ref{F:tHLM_DLC} we compare the behavior of the percentile DLC score for the anomaly vertices versus the vertices in the associated cohort.  Specifically, in Figures \ref{SF:gap} and \ref{SF:gap_CDF}, we show the distribution and the cumulative distribution, respectively, of the gap between the percentile DLC score with an anomaly present and the percentile score without the anomaly.  In Figures \ref{SF:evo} and \ref{SF:evo_CDF}, we instead compare the percentile DLC score with an anomaly with the score in the prior time step.  As Figure \ref{SF:gap_CDF} emphasizes, there is relatively little difference between the percentile scores for the cohort vertices with or without the anomaly present.  In contrast, from Figure \ref{SF:gap} the percentile scores for the vertices participating in the anomaly are typically about 2.5\% (for the leaves) to 5\% (for the root) higher than they would be without the anomaly.  In contrast, in Figures \ref{SF:evo} and \ref{SF:evo_CDF} we see that for both the anomaly vertices and the associated cohorts, the little typical change in percentile between the previous time step and the anomaly time step.  However, from Figure \ref{SF:evo} we can see that the vertices in the cohort of the anomaly have a significantly larger probability of having a sharp decrease in the percentile ranking.  Taken together, this indicates that there is likely a tendency to of the scores of vertices to naturally revert to the median behavior under the natural evolution, however an injected anomaly somewhat moderates this tendency in a localized way (i.e., while the vertices participating in the anomaly resist reverting towards the median, this effect does not carry over to similarly situated vertices).

\begin{figure}
    \centering
    \hfill
    \subfloat[Kernel Density Estimate for the distribution of the gap between the scores with and without the injected anomaly.]{\includegraphics[width=.45\linewidth]{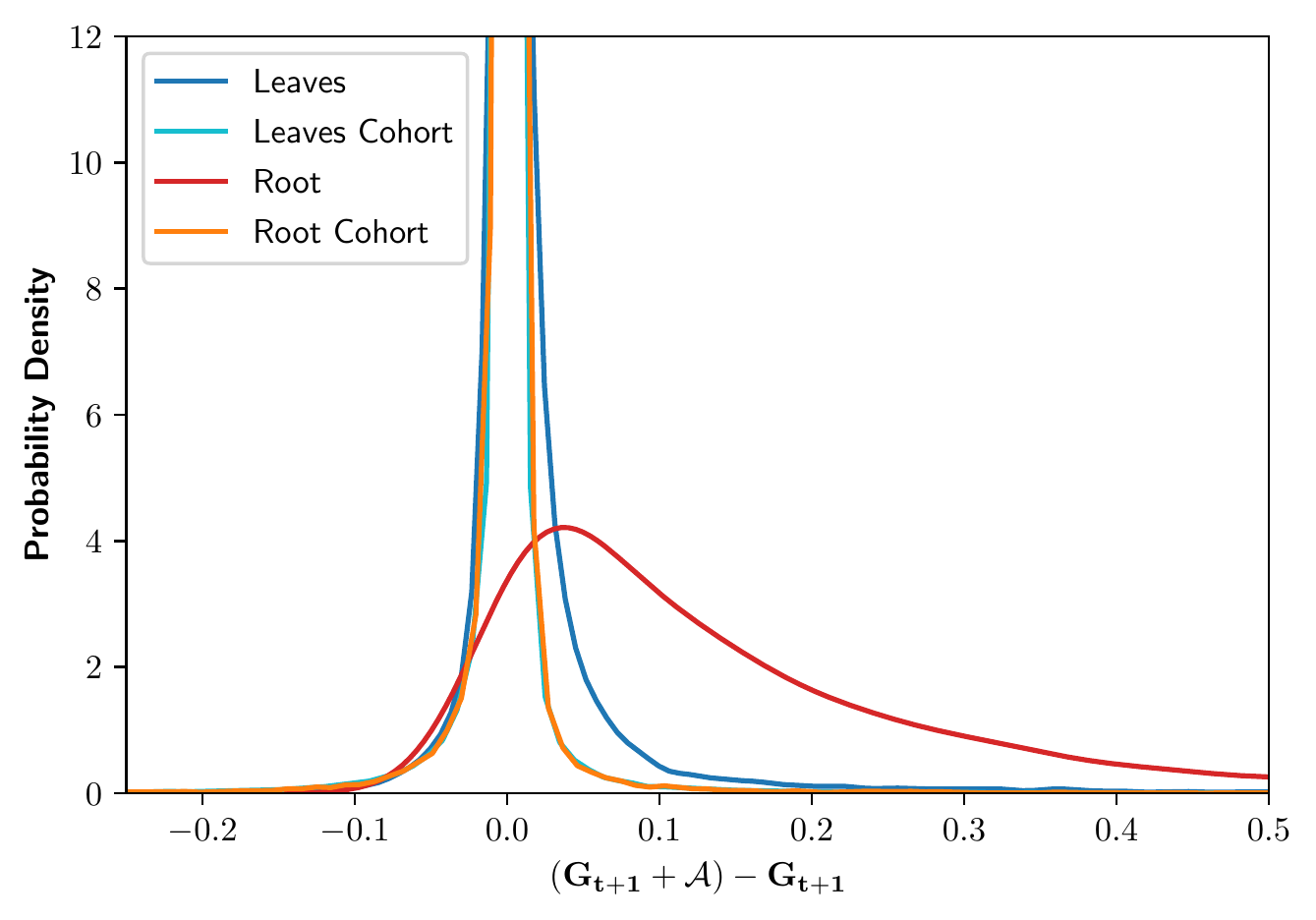}}
    \hfill
    \subfloat[Kernel Density Estimate for cumulative density function for the gap between the scores with and without the injected anomaly.]{\includegraphics[width=.45\linewidth]{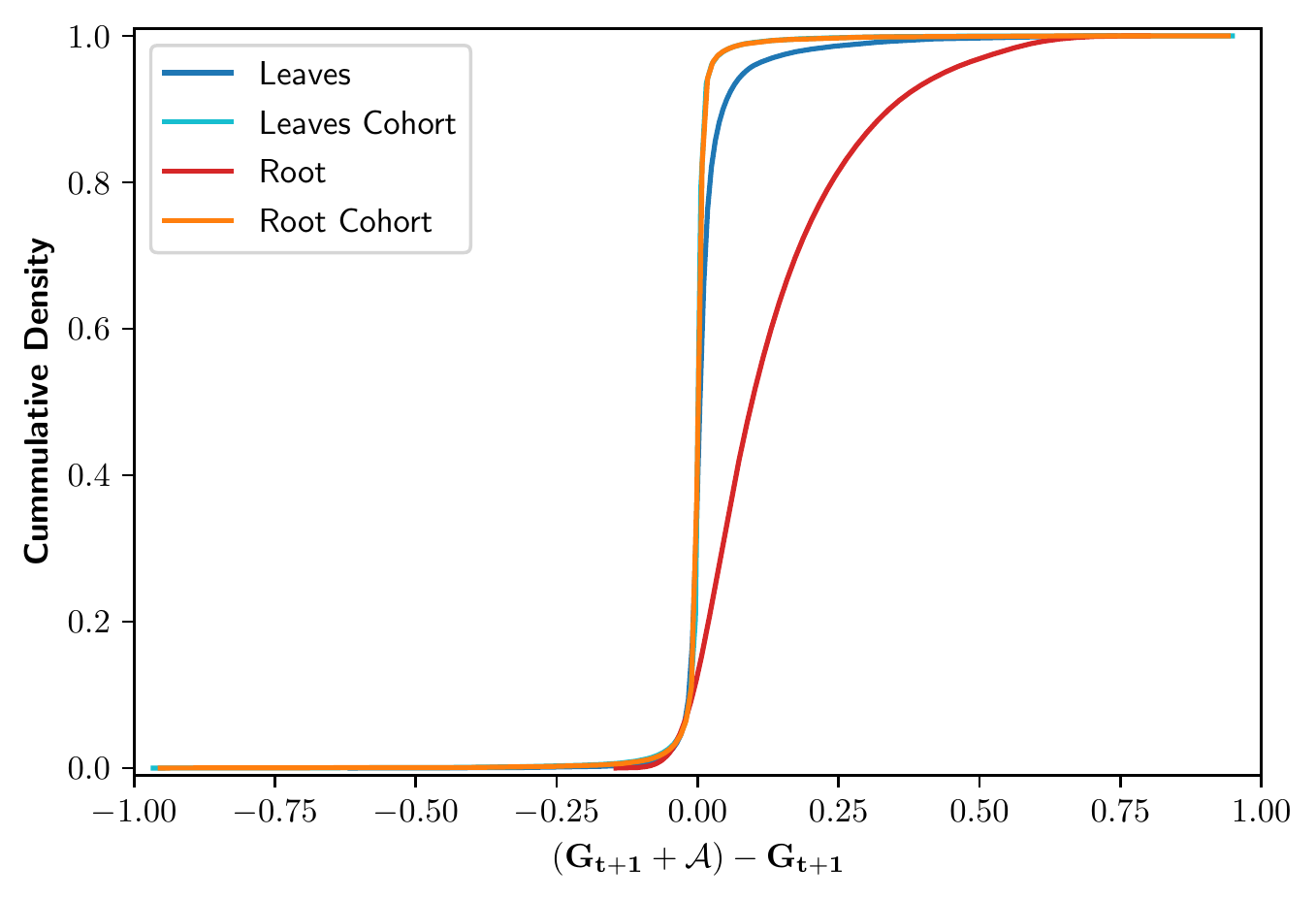}}
    \hfill\phantom{}\\
    \hfill
    \subfloat[Kernel Density Estimate for the distribution of the gap between the scores with the injected anomaly and the score in previous time step.]{\includegraphics[width=.45\linewidth]{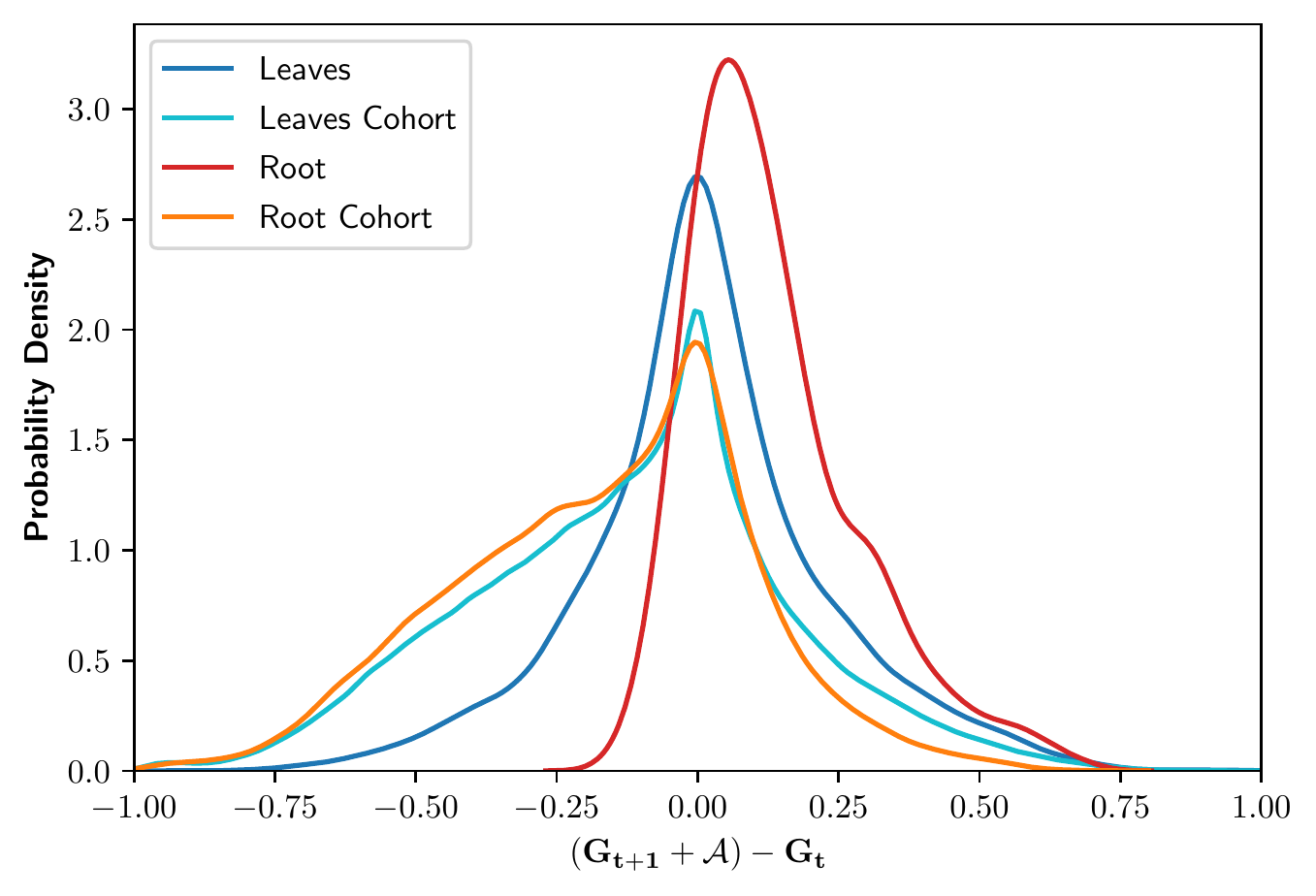}}
    \hfill
    \subfloat[Kernel Density Estimate for the cumulative density function of the gap between the scores with the injected anomaly and the score in previous time step.]{\includegraphics[width=.45\linewidth]{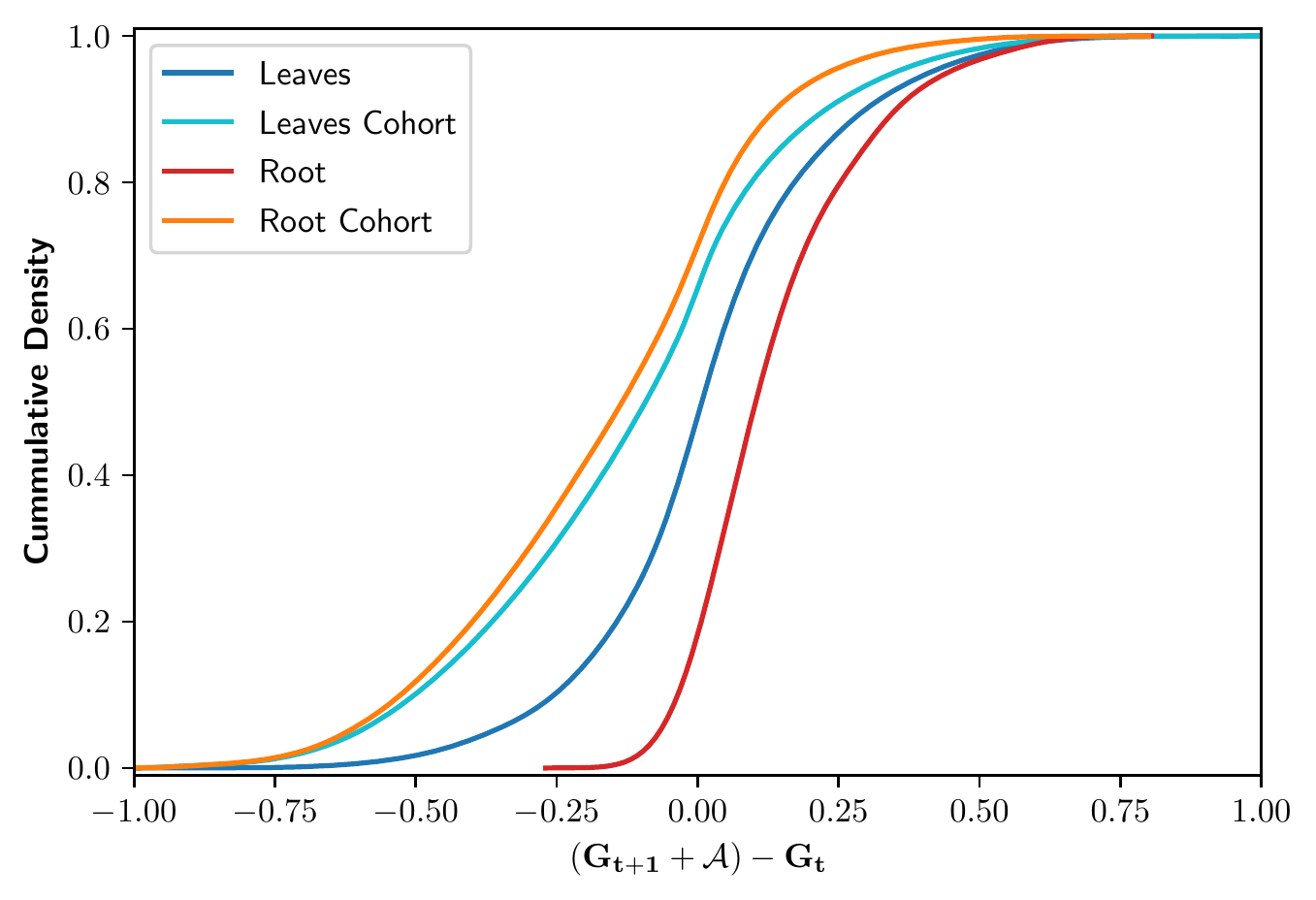}}
    \hfill\phantom{}
    \caption{$\bar{5}\mbox{-}\nDLC$}
    \label{F:tHLM_nDLC}
\end{figure}

The figures provided in Figure \ref{F:tHLM_nDLC} are analogous to those in Figure \ref{F:tHLM_DLC} except for being applied to the normalized DLC percentiles. We see that qualitatively, the results for the normalized DLC percentile are similar to those for the DLC percentile.  Quantitatively, Figure \ref{F:tHLM_nDLC} indicates that the  variability of for the normalized DLC percentiles is larger than that of the DLC, but the injection of the anomaly provides are larger boost to the percentile scores, particularly for the root of the anomaly. 

\section{Discussion}

This paper introduced two new spectral-based vertex centrality measures, DLC and nDLC, and used the setting of network flow data to illustrate their sensitivity and ability to identify anomalies planted in graphs.
Our new centralities are related to existing vertex centralities, most closely to the work of Qi, et al. \cite{qi2012laplacian}, but use a more nuanced approach by measuring changes to eigenvalues in the face of infinitesimal changes to a given vertex.
This allows for more subtle structural features in a graph to have measurable centrality.

Our experiments in Section \ref{sec:exp} show how the DLC and nDLC measures can be used to provide attribution of certain kinds of anomalies to cyber analysts.
In the first experiment we see that the percentiles of both DLC and nDLC, calculated using the top 5 eigenvalues, show significant increases when a star or clique anomaly is injected onto low importance vertices. 
This is true even when the injected anomaly is a rather small percentage of the graph. 
In the second experiment, rather than considering only low importance vertices, we inject a star anomaly of varying size randomly into the graph. 
Again, we see that even for small anomalies there is a measurable change in both DLC and nDLC percentiles for the root vertex of the star. 
These two experiments taken together show that on average, even small star anomalies manifest measurable changes in both centrality measures.
This indicates that looking for vertices that have a large increase in centrality from one time step to the next may narrow the set of IP addresses to investigate for malicious behavior.

Those first two experiments did not take into account temporal variation in the network. 
Rather, we compared a static, typical network flow graph $G$ to that same graph with an anomaly added.
In the third experiment we added the additional complexity of temporal variation. 
Using the \tHLM\ model to create a dynamic graph with properties similar to daily rhythm of network flow data, we injected a small star anomaly onto a set of vertices that are in what one could consider the {\it core} of the graph, i.e., vertices that are in the giant component for all time steps.
We showed that when comparing the DLC and nDLC of vertices in $G_{t+1}+\mathcal{A}$ to both $G_{t+1}$ and $G_t$ those vertices participating in the anomaly were more perturbed than their ``cohort'' vertices (those vertices having roughly the same centrality score in $G_t$).
This difference was more pronounced for nDLC than DLC, but still visually evident in both.
Given this observation we hypothesize that DLC and nDLC could be used to detect an anomaly of this type by tracking cohorts over time and identifying when vertices change cohorts.

Situational awareness comprises understanding the current posture of the network and predicting the security states of future time points. 
This paper provides a tool for identifying anomalies in a network in order to assess the current network state.
While our measure does not provide a high level overview of the network's security posture, the ability of these measures to point to even small anomalies may aid operators in preventing those from becoming large anomalies. 

\paragraph{\bf Acknowledgements} The authors would like to thank Helen Jenne for helpful discussions.
PNNL Information Release PNNL-SA-155008.

\bibliographystyle{plain}
\bibliography{wave}
\end{document}